\documentclass[conference,letterpaper]{IEEEtran}
\IEEEoverridecommandlockouts

\usepackage{cite}
\usepackage{url}
\usepackage{ifthen}
\usepackage{algorithm, algorithmicx, algpseudocode}
\usepackage{graphicx} 
\usepackage{textcomp}
\usepackage{booktabs}
\usepackage{makecell} 
\usepackage[cmex10]{amsmath}
\usepackage{amssymb,amsfonts}
\usepackage[left=1.62cm,right=1.62cm,top=1.85cm]{geometry}
\usepackage{pifont}
\usepackage{amsthm}
\usepackage{mathrsfs}
\def\BibTeX{{\rm B\kern-.05em{\sc i\kern-.025em b}\kern-.08em T\kern-.1667em\lower.7ex\hbox{E}\kern-.125emX}}
\usepackage{tikz}
\usetikzlibrary{automata,positioning,calc}
\usetikzlibrary{intersections}
\usetikzlibrary{decorations.pathreplacing,angles,quotes}

\usepackage{bm}
\usepackage{amscd}

\algnewcommand{\Initialize}[1]{%
  \State \textbf{Initialization:}
  \Statex \hspace*{\algorithmicindent}\parbox[t]{0.8\linewidth}{\raggedright #1}
}
\makeatletter

\theoremstyle{definition}
\newtheorem{theorem}{Theorem}
\newtheorem{definition}{Definition}
\newtheorem{lemma}{Lemma}
\newtheorem{prop}{Proposition}
\newtheorem{cor}{Corollary}
\newtheorem{remark}{Remark}

\def\bn{\mathbb N}

\def\vE{\mathbb E}

\font\b=cmr10 scaled\magstep4

\def\bigzerou{\smash{\lower1.7ex\hbox{\b 0}}}
\def\bigzerou{\smash{\lower1.7ex\hbox{\b 0}}}

\begin{document}
\title{Alternating Optimization Approach for Computing $\alpha$-Mutual Information and $\alpha$-Capacity
\thanks{This work was supported by JSPS KAKENHI Grant Number JP23K16886.}
}

\author{
\IEEEauthorblockN{Akira Kamatsuka}
\IEEEauthorblockA{Shonan Institute of Technology \\ 
Email: \text{kamatsuka@info.shonan-it.ac.jp}
 }
\and
\IEEEauthorblockN{Koki Kazama}
\IEEEauthorblockA{Shonan Institute of Technology \\ 
Email: \text{kazama@info.shonan-it.ac.jp}
 }
\and
\IEEEauthorblockN{Takahiro Yoshida}
\IEEEauthorblockA{Nihon University \\ 
Email: \text{yoshida.takahiro@nihon-u.ac.jp}
 } 
}
\maketitle
\begin{abstract}
This study presents alternating optimization (AO) algorithms for computing $\alpha$-mutual information ($\alpha$-MI) and $\alpha$-capacity based on variational characterizations of $\alpha$-MI using a reverse channel. 
Specifically, we derive several variational characterizations of Sibson, Arimoto, Augustin--Csisz{\' a}r, and Lapidoth--Pfister MI and introduce novel AO algorithms for computing $\alpha$-MI and $\alpha$-capacity; their performances for computing $\alpha$-capacity are also compared.
The comparison results show that the AO algorithm based on the Sibson MI's characterization has the fastest convergence speed. 
\end{abstract}

\section{Introduction}
Sibson mutual information (MI) $I_{\alpha}^{\text{S}}(X; Y)$ \cite{Sibson1969InformationR}, Arimoto MI $I_{\alpha}^{\text{A}}(X; Y)$ \cite{arimoto1977}, 
Augustin--Csisz{\' a}r MI $I_{\alpha}^{\text{C}}(X; Y)$ \cite{augusting_phd_thesis},\cite{370121}, and Lapidoth--Pfister MI $I_{\alpha}^{\text{LP}}(X; Y)$ \cite{e21080778},\cite{8231191} 
are well-known generalizations of the Shannon MI $I(X; Y)$ \cite{shannon} using the tunable parameter $\alpha\in (0, 1)\cup (1, \infty)$.
These types of MI are referred to as $\alpha$-mutual information ($\alpha$-MI) \cite{7308959} and defined based on the R{\' e}nyi entropy $H_{\alpha}(X)$, Arimoto conditional entropy $H_{\alpha}^{\text{A}}(X | Y)$, and R{\' e}nyi divergence $D_{\alpha}(\cdot || \cdot)$. 
The maximum value of the $\alpha$-MI for the input distribution $p_{X}$, given by $C_{\alpha}^{(\cdot)}:=\max_{p_{X}} I_{\alpha}^{(\cdot)}(X; Y)$, is referred to as \textit{$\alpha$-capacity}.
The $\alpha$-MI and $\alpha$-capacity characterize various problems in information theory, such as 
the error and correct decoding exponents in channel coding 
\cite{Gallager:1968:ITR:578869,1055007,OMURA1975148,8423117,Nakiboglu:2019aa,e23020199}, the generalized cutoff rate in channel coding \cite{370121}, 
the error and strong converse exponents in composite hypothesis testing \cite{8231191}, and the (maximal) $\alpha$-leakage in privacy-preserving data publishing \cite{8804205}. 

While the Sibson and Arimoto MI have closed-form expressions, the Augustin--Csisz{\' a}r MI, Lapidoth--Pfister MI, and $\alpha$-capacity generally do not; therefore, these types of MI and $\alpha$-capacity must be calculated by solving optimization problems. 
Algorithms for computing the $\alpha$-MI and $\alpha$-capacity have been explored in \cite{arimoto1977,BN01990060en,10619200,augusting_phd_thesis,4595361,9611513,9834648,10619174,10619657,kamatsuka2024_ISITA}.
For the Augustin--Csisz{\' a}r MI for $\alpha\in (0, 1)$, Augustin \cite{augusting_phd_thesis} proposed an iterative algorithm, which Karakos \textit{et al.} \cite{4595361} rediscovered as an alternating optimization (AO) algorithm.
Recently, other iterative algorithms have been proposed in \cite{9611513,9834648,10619174,10619657}, which are effective for $\alpha \in (1, \infty)$.
For the Sibson and Arimoto capacities, Arimoto \cite{arimoto1977,BN01990060en} proposed AO algorithms similar to the Arimoto--Blahut algorithm \cite{1054753,1054855} for calculating the channel capacity $C$. 
Recently, Kamatsuka \textit{et al.} \cite{10619200} proposed additional AO algorithms for computing these capacities and demonstrated that they are equivalent to those suggested by Arimoto.
For Augustin--Csisz{\' a}r capacity, Kamatsuka \textit{et al.} \cite{kamatsuka2024_ISITA} proposed an AO algorithm.
Table \ref{tab:various_alpha_MI} summarizes the $\alpha$-MI, $\alpha$-capacity, and their calculation methods.

\begin{table}[t]
  \caption{Examples of $\alpha$-MI and their calculation methods}
  \label{tab:various_alpha_MI}
  \centering
  \resizebox{.49\textwidth}{!}{
  \begin{tabular}{@{} ccc @{}}
    \toprule
    Name & Calculation method \\ 
    \midrule
    \begin{tabular}{c}
    Sibson MI  \\ 
    $I_{\alpha}^{\text{S}}(X; Y)$ \\ 
    Sibson capacity \\
    $C_{\alpha}^{\text{S}}$
    \end{tabular}
    & \begin{tabular}{c}
    Closed-form expression 
    \ \\ \ \\
    AO algorithms \cite{arimoto1977,10619200}, (\textbf{This work})
    \end{tabular}  
    \\ 
    \midrule
    \begin{tabular}{c}
    Arimoto MI \\ 
    $I_{\alpha}^{\text{A}}(X; Y)$ \\
    Arimoto capacity \\ 
    $C_{\alpha}^{\text{A}}$
    \end{tabular}
    & \begin{tabular}{c}
    Closed-form expression  
    \ \\ \ \\ 
    AO algorithms \cite{BN01990060en,10619200}
    \end{tabular} 
    \\ 
    \midrule
    \begin{tabular}{c}
    Augustin--Csisz{\' a}r MI \\ 
    $I_{\alpha}^{\text{C}}(X; Y)$ \\ \ \\
    Augustin--Csisz{\' a}r capacity \\ 
    $C_{\alpha}^{\text{C}}$
    \end{tabular}
    & \begin{tabular}{c}
    AO algorithm for $\alpha\in (0, 1)$ \cite{augusting_phd_thesis,4595361} \\ 
    AO algorithm for $\alpha\in (1, \infty)$ (\textbf{This work}) \\ 
    Iterative algorithms \cite{9611513,9834648,10619174,10619657}
    \ \\ \ \\
    AO algorithm for $\alpha\in (1, \infty)$ \cite{kamatsuka2024_ISITA}
    \end{tabular} 
    \\ 
    \midrule
    \begin{tabular}{c}
    Lapidoth--Pfister MI \\ 
    $I_{\alpha}^{\text{LP}}(X; Y)$ \\ 
    Lapidoth--Pfister capacity \\ 
    $C_{\alpha}^{\text{LP}}$
    \end{tabular}
    & 
    \begin{tabular}{c}
    AO algorithms (\textbf{This work}) \\ 
    \ \\ 
    AO algorithm for $\alpha\in (1, \infty)$ (\textbf{This work}) \\ 
    \end{tabular}
    \\ 
    \bottomrule
  \end{tabular}
  }
\end{table}

These AO algorithms are based on the variational characterization of $\alpha$-MI, transforming the definitions of $\alpha$-MI into optimization problems.
This study proposes novel variational characterizations of $\alpha$-MI and presents AO algorithms based on these characterizations. 
The main contributions of this paper are as follows:
\begin{itemize}
\item Extending the results of the variational characterization of the Sibson MI by Shayevitz \cite{6034266}, we provide a novel proof for Arimoto's variational characterization of the Sibson MI using the reverse channel $r_{X\mid Y}$ \cite{arimoto1977} (Theorem \ref{thm:variational_characterization_Sibson_MI}). 
As a byproduct of this proof, we also present a new characterization of the Sibson MI.
Similarly, we provide novel variational characterizations of the Augustin--Csisz{\' a}r and Lapidoth--Pfister MI (Theorems \ref{thm:variational_characterization_Csiszar_MI} and \ref{thm:variational_characterization_LP_MI}).
\item Based on the characterizations, we propose novel AO algorithms for computing Augustin--Csisz{\' a}r MI, Lapidoth--Pfister MI, and Sibson and Lapidoth--Pfister capacities (Section \ref{sec:AO}).
\item We compare the performance of the AO algorithms for computing $\alpha$-capacity through numerical examples (Section \ref{ssec:comparison}). 
\end{itemize}

\section{Preliminaries}\label{sec:preliminaries}
Let $(X, Y)$ denote random variables distributed according to a joint distribution $p_{X, Y} = p_{X}p_{Y\mid X}$ on finite alphabets $\mathcal{X}\times \mathcal{Y}$.
Let $H(X) := - \sum_{x}p_{X}(x)\log p_{X}(x)$ be the Shannon entropy, $H(X|Y) := -\sum_{y}p_{Y}(y)\sum_{x}p_{X\mid Y}(x | y)\log p_{X\mid Y}(x|y)$ the conditional entropy, 
and $I(X; Y) := H(X) - H(X | Y)$ the Shannon MI, where $p_{Y}$ is the marginal distribution of $Y$. 
We denote the expectation of $f(X)$ as $\vE_{X}^{p_{X}}\left[f(X)\right]:=\sum_{x}p_{X}(x)f(x)$.
Throughout this paper, we use $\log$ to represent the natural logarithm.

We initially review the Shannon MI and $\alpha$-MI. 
The Shannon MI $I(X; Y)=I(p_{X}, p_{Y\mid X})$ have the following variational characterizations.
\begin{prop}[\text{\cite[Cor 4.2, Thm 4.3 and 4.4]{Polyanskiy_Wu_2024}}] \label{prop:vc_Shannon_MI}
\begin{align}
I(p_{X}, p_{Y\mid X}) 
&= \min_{q_{Y}} D(p_{X}p_{Y\mid X} || p_{X}q_{Y}) \label{eq:min_characterization_Shannon_MI} \\ 
&= \min_{q_{X}}\min_{q_{Y}} D(p_{X}p_{Y\mid X} || q_{X}q_{Y}) \label{eq:minmin_characterization_Shannon_MI} \\ 
&= \max_{r_{X\mid Y}} \vE_{X, Y}^{p_{X}p_{Y\mid X}} \log \left[\frac{r_{X\mid Y}(X\mid Y)}{p_{X}(X)}\right], \label{eq:max_characterization_Shannon_MI}
\end{align}
where the minimum in \eqref{eq:min_characterization_Shannon_MI} is taken over all distributions on $\mathcal{Y}$ and is achieved at $q_{Y}^{*} = p_{Y}$. 
The minimum in \eqref{eq:minmin_characterization_Shannon_MI} is taken over all product distributions on $\mathcal{X}\times \mathcal{Y}$ and is achieved at $(q_{X}^{*}, q_{Y}^{*}) = (p_{X}, p_{Y})$.
The maximum in \eqref{eq:max_characterization_Shannon_MI} is taken over all the family of conditional distributions, i.e., reverse channels $r_{X\mid Y}:=\{r_{X\mid Y}(\cdot | y)\}_{y\in \mathcal{Y}}$
and is achieved at $r^{*}_{X\mid Y}= p_{X\mid Y}$. 
\end{prop}

\begin{definition}
Let $\alpha\in (0, 1)\cup (1, \infty)$ and $(X, Y)\sim p_{X, Y}=p_{X}p_{Y\mid X}$. 
Then, the \textit{Sibson MI of order $\alpha$}, \textit{Arimoto MI of order $\alpha$}, 
\textit{Augustin--Csisz{\'a}r MI of order $\alpha$}, and \textit{Lapidoth--Pfister MI of order $\alpha$} are  defined as follows, respectively:
\begin{align}
I_{\alpha}^{\text{S}} (X; Y) &:= \min_{q_{Y}} D_{\alpha} (p_{X}p_{Y\mid X} || p_{X}q_{Y}), \label{eq:def_Sibson_MI} \\ 
I_{\alpha}^{\text{A}}(X; Y) &:= H_{\alpha}(X) - H_{\alpha}^{\text{A}}(X\mid Y), \label{eq:def_Arimoto_MI} \\ 
I_{\alpha}^{\text{C}}(X; Y) &:= \min_{q_{Y}} \vE_{X}^{p_{X}}\left[D_{\alpha}(p_{Y\mid X}(\cdot \mid X) || q_{Y})\right], \label{eq:def_AC_MI} \\ 
I_{\alpha}^{\text{LP}}(X; Y) &:= \min_{q_{X}}\min_{q_{Y}} D_{\alpha}(p_{X}p_{Y\mid X} || q_{X}q_{Y}), \label{eq:def_LP_MI}
\end{align}
where 
\begin{align}
D_{\alpha}(p || q) := \frac{1}{\alpha-1}\log \sum_{z} p(z)^{\alpha}q(z)^{1-\alpha} \label{eq:def_Renyi_divergence}
\end{align}
is the R{\' e}nyi divergence between $p$ and $q$ of order $\alpha$, 
$H_{\alpha}(X) := \frac{1}{1-\alpha} \log \sum_{x} p_{X}(x)^{\alpha}$ is the R{\' e}nyi entropy of order $\alpha$, 
and $H_{\alpha}^{\text{A}}(X | Y):= \frac{\alpha}{1-\alpha}\log\sum_{y} \left( \sum_{x}p_{X}(x)^{\alpha}p_{Y\mid X}(y | x)^{\alpha} \right)^{\frac{1}{\alpha}}$ is the Arimoto conditional entropy of order $\alpha$ \cite{arimoto1977}.
\end{definition}

\begin{remark}
The Sibson and Arimoto MI have the following closed-form expressions using the Gallager error exponent function $E_{0}(\rho, p_{X}):= -\log \sum_{y}\left( \sum_{x}p_{X}(x)p_{Y\mid X}(y | x)^{\frac{1}{1+\rho}} \right)^{1+\rho}$ \cite{Gallager:1968:ITR:578869}:  
\begin{align}
I_{\alpha}^{\text{S}} (X; Y) 
&=  \frac{\alpha}{1-\alpha} E_{0} \left( \frac{1}{\alpha}-1, p_{X} \right), \label{eq:Sibson_MI_Gallager} \\ 
I_{\alpha}^{\text{A}}(X; Y) 
&= \frac{\alpha}{1-\alpha} E_{0}\left( \frac{1}{\alpha}-1, p_{X_{\alpha}} \right), \label{eq:Arimoto_MI_Gallager}
\end{align}
where $p_{X_{\alpha}}(x) := \frac{p_{X}(x)^{\alpha}}{\sum_{x}p_{X}(x)^{\alpha}}$ is the $\alpha$-tilted distribution \cite{8804205} 
(also known as the escort distribution \cite{10.5555/3019383}) of $p_{X}$. 
However, the closed-form expressions of the Augustin--Csisz{\'a}r MI and Lapidoth--Pfister MI are unknown and must be computed by solving the optimization problems.
\end{remark}

Csisz{\' a}r \cite{370121}, Karakos \textit{et al.} \cite{4595361}, Shayevitz \cite{6034266}, and Lapidoth and Pfister \cite{e21080778} provided the following variational characterizations of $\alpha$-MI.

\begin{prop}[\text{\cite[p. 34]{370121}, \cite{4595361}, \cite[Thm 1]{6034266}, \cite[Lemma 8]{e21080778}}]
Let $\alpha\in (0, 1)\cup (1, \infty)$. Then, 
\begin{align}
&(1-\alpha)I_{\alpha}^{\text{C}}(X; Y) \notag \\ 
&= \min_{\tilde{q}_{Y\mid X}} \left\{\alpha D(p_{X}\tilde{q}_{Y\mid X} || p_{X}p_{Y\mid X}) + (1-\alpha)I(p_{X}, \tilde{q}_{Y\mid X})\right\}, \label{eq:vc_AC_MI_KL} \\ 
&(1-\alpha)I_{\alpha}^{\text{S}}(X; Y) \notag \\ 
&= \min_{\tilde{q}_{X, Y}} \{\alpha D(\tilde{q}_{X, Y} || \tilde{q}_{X}p_{Y\mid X}) \notag \\ 
&\qquad \qquad \qquad  + (1-\alpha)I(\tilde{q}_{X}, \tilde{q}_{Y\mid X}) + D(\tilde{q}_{X} || p_{X})\}, \label{eq:vc_Sibson_MI_KL} \\
&(1-\alpha)I_{\alpha}^{\text{LP}}(X; Y) = \min_{\tilde{q}_{X, Y}} \{\alpha D(\tilde{q}_{X, Y} || p_{X}p_{Y\mid X}) \notag \\ 
&\qquad \qquad \qquad \qquad \qquad +  (1-\alpha) I(\tilde{q}_{X}, \tilde{q}_{Y\mid X})\}. \label{eq:vc_LP_MI_KL}
\end{align}
\end{prop}

\section{Variational Characterizations of $\alpha$-MI}\label{sec:main_result}

It is natural to consider variational characterizations of $\alpha$-MI similar to the Shannon MI characterization in Proposition \ref{prop:vc_Shannon_MI}, 
some of which were already obtained by Arimoto \cite{arimoto1977,BN01990060en}, Kamatsuka \textit{et al.} \cite{10619200}, and Karakos \textit{et al.} \cite{4595361}. 
In this section, we derive new characterizations of Sibson, Augustin--Csisz{\' a}r, and Lapidoth--Pfister MI through a novel proof of Arimoto's characterization.

\begin{theorem}[\text{\cite[Thm 4]{arimoto1977},\cite[Thm 1]{10619200}}] \label{thm:variational_characterization_Sibson_MI}
Let $\alpha \in (0, 1)\cup (1, \infty)$. Then, 
\begin{align}
I_{\alpha}^{\text{S}} (X; Y) &= \max_{r_{X\mid Y}} F_{\alpha}^{\text{S1}}(p_{X}, r_{X\mid Y}) \label{eq:vc_Sibson_MI_01} \\ 
&= \max_{r_{X\mid Y}} F_{\alpha}^{\text{S2}}(p_{X}, r_{X\mid Y}), \label{eq:vc_Sibson_MI_02}
\end{align}
where 
\begin{align}
&F_{\alpha}^{\text{S1}}(p_{X}, r_{X\mid Y}) \notag \\ 
&:= \frac{\alpha}{\alpha-1}\log \sum_{x, y} p_{X}(x)^{\frac{1}{\alpha}}p_{Y\mid X}(y | x)r_{X\mid Y}(x | y)^{1-\frac{1}{\alpha}}, \label{eq:def_vc_Sibson_MI_01} \\ 
&F_{\alpha}^{\text{S2}}(p_{X}, r_{X\mid Y}) := F_{\alpha}^{\text{S1}}(p_{X}, r_{X_{\alpha}\mid Y}) \notag \\ 
&= \frac{\alpha}{\alpha-1}\log \sum_{x, y} p_{X}(x)^{\frac{1}{\alpha}}p_{Y\mid X}(y | x)r_{X_{\alpha}\mid Y}(x | y)^{1-\frac{1}{\alpha}}. \label{eq:def_vc_Sibson_MI_02}
\end{align}
For $\alpha \in (1, \infty)$, 
\begin{align}
I_{\alpha}^{\text{S}} (X; Y) &= \max_{\tilde{q}_{X, Y}}\max_{r_{X\mid Y}} \tilde{F}_{\alpha}^{\text{S3}}(p_{X}, \tilde{q}_{X, Y}, r_{X\mid Y}), \label{eq:vc_Sibson_MI_03}
\end{align}
where 
\begin{align}
&\tilde{F}_{\alpha}^{\text{S3}}(p_{X}, \tilde{q}_{X, Y}, r_{X\mid Y}) :=\frac{\alpha}{1-\alpha}D(\tilde{q}_{X, Y} || \tilde{q}_{X}p_{Y\mid X}) \notag \\ 
&+ \vE_{X, Y}^{\tilde{q}_{X, Y}}\left[\log \frac{r_{X\mid Y}(X|Y)}{\tilde{q}_{X}(X)}\right] + \frac{1}{1-\alpha} D(\tilde{q}_{X} || p_{X}). \label{eq:def_vc_Sibson_MI_03}
\end{align}
\end{theorem}
\begin{proof}
See Appendix \ref{proof:variational_characterization_Sibson_MI}. 
\end{proof}

\begin{remark} \label{rem:vc_Arimoto}
From \eqref{eq:Sibson_MI_Gallager} and \eqref{eq:Arimoto_MI_Gallager}, similar characterizations of the Arimoto MI can be obtained by replacing $p_{X}$ in \eqref{eq:vc_Sibson_MI_01}, \eqref{eq:vc_Sibson_MI_02}, and \eqref{eq:vc_Sibson_MI_03} with $p_{X_{\alpha}}$, two of which were first obtained by Arimoto \cite{BN01990060en} and Kamatsuka \textit{et al.} \cite{10619200}.
\end{remark}

Similar to the proof for the Sibson MI, the following variational characterizations of the Augustin--Csisz{\'a}r MI and Lapidoth--Pfister MI are obtained.

\begin{theorem} \label{thm:variational_characterization_Csiszar_MI}
Let $\alpha \in (0, 1)\cup (1, \infty)$. Then, 
\begin{align}
&I_{\alpha}^{\text{C}}(X; Y) \notag \\
&= \begin{cases}
\min_{\tilde{q}_{Y\mid X}}\min_{q_{Y}} F_{\alpha}^{\text{C}}(\tilde{q}_{Y\mid X}, q_{Y}), & \alpha \in (0, 1), \\ 
\max_{\tilde{q}_{Y\mid X}}\max_{r_{X\mid Y}} \tilde{F}_{\alpha}^{\text{C}}(\tilde{q}_{Y\mid X}, r_{X\mid Y}), & \alpha \in (1, \infty) \label{eq:vc_AC_MI_02_03}
\end{cases} \\ 
&= H(p_{X}) \notag \\ 
&+ \max_{r_{X\mid Y}} \frac{\alpha}{\alpha-1}\sum_{x}p_{X}(x) \log \sum_{y}p_{Y\mid X}(y | x)  r_{X\mid Y}(x | y)^{1-\frac{1}{\alpha}}, \label{eq:vc_AC_MI_01} 
\end{align}
where 
\begin{align}
&{F}_{\alpha}^{\text{C}}(\tilde{q}_{Y\mid X}, q_{Y}) \notag \\ 
&= \frac{\alpha}{1-\alpha} D(p_{X}\tilde{q}_{Y\mid X} || p_{X}p_{Y\mid X}) + D(p_{X}\tilde{q}_{Y\mid X} || p_{X}q_{Y}) \label{eq:def_vc_AC_MI_02} \\
&\tilde{F}_{\alpha}^{\text{C}}(\tilde{q}_{Y\mid X}, r_{X\mid Y}) \notag \\ 
&:=\frac{\alpha}{1-\alpha} D(p_{X}\tilde{q}_{Y\mid X} || p_{X}p_{Y\mid X}) + \vE_{X, Y}^{p_{X}\tilde{q}_{Y\mid X}}\left[\frac{r_{X\mid Y}(X\mid Y)}{p_{X}(X)}\right]. \label{eq:def_vc_AC_MI_03}
\end{align}
\end{theorem}
\begin{proof}
See Appendix \ref{proof:variational_characterization_Csiszar_MI}. 
\end{proof}

\begin{remark}
It is worth mentioning that Poltyrev \cite[Thm 3.4]{poltyrev1982random} reported the equivalent characterization in \eqref{eq:vc_AC_MI_01} (see \cite[Lemma 18]{Nakiboglu:2019aa}). 
\end{remark}

\begin{theorem} \label{thm:variational_characterization_LP_MI}
Let $\alpha \in (0, 1) \cup (1, \infty)$. Then, 
\begin{align}
&I_{\alpha}^{\text{LP}}(X; Y) \notag \\ 
&= 
\begin{cases}
\min_{\tilde{q}_{X, Y}}\min_{q_{Y}} F_{\alpha}^{\text{LP}}(\tilde{q}_{X, Y}, q_{Y}), & \alpha \in (0, 1),  \\ 
\max_{\tilde{q}_{X, Y}}\max_{r_{X\mid Y}} \tilde{F}_{\alpha}^{\text{LP}}(\tilde{q}_{X, Y}, r_{X\mid Y}), & \alpha \in (1, \infty), 
\end{cases}\label{eq:vc_LP_MI_02_03}
\end{align}
where $\tilde{q}_{X, Y} = \tilde{q}_{X}\tilde{q}_{Y\mid X}$ and 
\begin{align}
&F_{\alpha}^{\text{LP}}(\tilde{q}_{X, Y}, q_{Y}) := \frac{\alpha}{1-\alpha} D(\tilde{q}_{X, Y} || \tilde{q}_{X} p_{Y\mid X}) \notag \\ 
&+ D(\tilde{q}_{X}\tilde{q}_{Y\mid X} || \tilde{q}_{X}q_{Y}) + \frac{\alpha}{1-\alpha} D(\tilde{q}_{X} || p_{X}), \label{eq:def_vc_LP_MI_02} \\ 
&\tilde{F}_{\alpha}^{\text{LP}}(\tilde{q}_{X, Y}, r_{X\mid Y}) := \frac{\alpha}{1-\alpha} D(\tilde{q}_{X, Y} || \tilde{q}_{X} p_{Y\mid X}) \notag \\ 
&+ \vE_{X, Y}^{\tilde{q}_{X, Y}}\left[\log \frac{r_{X\mid Y}(X|Y)}{\tilde{q}_{X}(X)}\right] + \frac{\alpha}{1-\alpha} D(\tilde{q}_{X} || p_{X}). \label{eq:def_vc_LP_MI_03}
\end{align}
For $\alpha \in (1, \infty)$, 
\begin{align}
&I_{\alpha}^{\text{LP}}(X; Y) 
= \max_{r_{X\mid Y}} \frac{2\alpha-1}{\alpha-1}  \notag \\ 
&\times \log \sum_{x}p_{X}(x)^{\frac{\alpha}{2\alpha-1}} \left( \sum_{y} p_{Y\mid X}(y | x) r_{X\mid Y}(x | y)^{1-\frac{1}{\alpha}} \right)^{\frac{\alpha}{2\alpha-1}}. \label{eq:vc_LP_MI_01}  
\end{align}

\end{theorem}
\begin{proof}
See Appendix \ref{proof:variational_characterization_LP_MI}. 
\end{proof}

\section{AO Algorithms} \label{sec:AO}

This section proposes novel AO algorithms for computing Augustin--Csisz{\' a}r MI, Lapidoth--Pfister MI, and $\alpha$-capacities 
based on the variational characterizations obtained in Section \ref{sec:main_result}.

\subsection{AO Algorithms for Computing the Augustin--Csisz{\' a}r MI}
Karakos \textit{et al.} \cite{4595361} proposed an AO algorithm based on the variational characterization in \eqref{eq:vc_AC_MI_02_03} for the Augustin--Csisz{\' a}r MI when $\alpha \in (0, 1)$. 
We propose an AO algorithm for $\alpha \in (1, \infty)$ with the following updating formulae:

\begin{prop} \label{prop:update_formulae_Csiszar_MI}
Let $\alpha\in (1, \infty)$. Then, the following holds:
\begin{enumerate}
\item For a fixed $\tilde{q}_{Y\mid X}$, $\tilde{F}_{\alpha}^{\text{C}}(\tilde{q}_{Y\mid X}, r_{X\mid Y})$ is maximized by 
\begin{align}
r_{X\mid Y}^{*}(x\mid y) &:= \frac{p_{X}(x)\tilde{q}_{Y\mid X}(y\mid x)}{\sum_{x} p_{X}(x)\tilde{q}_{Y\mid X}(y\mid x)}. \label{eq:C_update_r}
\end{align}
\item For a fixed $r_{X\mid Y}$, $\tilde{F}_{\alpha}^{\text{C}}(\tilde{q}_{Y\mid X}, r_{X\mid Y})$ is maximized by 
\begin{align}
\tilde{q}_{Y\mid X}^{*}(y\mid x) &:= \frac{p_{Y\mid X}(y | x) r_{X\mid Y}(x | y)^{1-\frac{1}{\alpha}}}{\sum_{y} p_{Y\mid X}(y | x) r_{X\mid Y}(x | y)^{1-\frac{1}{\alpha}}}. \label{eq:C_update_q}
\end{align}
\end{enumerate}
\end{prop}
\begin{proof}
This follows from the proof of Theorem \ref{thm:variational_characterization_Csiszar_MI} and Proposition \ref{prop:vc_Shannon_MI}.
\end{proof}

\subsection{AO Algorithms for Computing the Lapidoth--Pfister MI}

The Lapidoth--Pfister MI in \eqref{eq:def_LP_MI} is defined as a double minimization problem; $I_{\alpha}^{\text{LP}}(X; Y):=\min_{q_{X}}\min_{q_{Y}}D_{\alpha}(p_{X}p_{Y\mid X} || q_{X}q_{Y})$. 
From the definition and properties of the Lapidoth--Pfister MI (\cite[Prop 9 and Lemma 16]{e21080778}, \cite[Lemma 29]{8231191}), we derive an AO algorithm for the Lapidoth--Pfister MI with the following updating formulae.

\begin{cor} \label{cor:update_formulae_LP_MI}
Let $\alpha\in (0, 1) \cup (1, \infty)$ and $p_{X, Y}=p_{X}p_{Y\mid X}$. Then, the following holds:
\begin{enumerate}
\item For a fixed $q_{Y}$, $D_{\alpha}(p_{X}p_{Y\mid X} || q_{X}q_{Y})$ is minimized by 
\begin{align}
q_{X}^{*}(x) &:= \frac{\left[\sum_{y}p_{X, Y}(x, y)^{\alpha}q_{Y}(y)^{1-\alpha}\right]^{\frac{1}{\alpha}}}{\sum_{x}\left[\sum_{y}p_{X, Y}(x, y)^{\alpha}q_{Y}(y)^{1-\alpha}\right]^{\frac{1}{\alpha}}}.
\end{align}
\item For a fixed $q_{X}$, $D_{\alpha}(p_{X}p_{Y\mid X} || q_{X}q_{Y})$ is minimized by 
\begin{align}
q_{Y}^{*}(y) &:= \frac{\left[\sum_{x}p_{X, Y}(x, y)^{\alpha}q_{X}(x)^{1-\alpha}\right]^{\frac{1}{\alpha}}}{\sum_{y}\left[\sum_{x}p_{X, Y}(x, y)^{\alpha}q_{X}(x)^{1-\alpha}\right]^{\frac{1}{\alpha}}}.
\end{align}
\end{enumerate}
\end{cor}

Additional AO algorithm can be derived based on the Lapidoth--Pfister MI characterization in \eqref{eq:def_vc_LP_MI_02}.
\begin{prop} \label{prop:update_formulae_LP_MI_vc}
Let $F_{\alpha}^{\text{LP}}(\tilde{q}_{X, Y}, q_{Y})$ and $\tilde{F}_{\alpha}^{\text{LP}}(\tilde{q}_{X, Y}, r_{X\mid Y})$ be functionals defined in \eqref{eq:def_vc_LP_MI_02} and \eqref{eq:def_vc_LP_MI_03}, respectively.

Then, for $\alpha \in (0, 1)$, the following holds:
\begin{enumerate}
\item For a fixed $\tilde{q}_{X, Y}$, $F_{\alpha}^{\text{LP}}(\tilde{q}_{X, Y}, q_{Y})$ is minimized by 
\begin{align}
q_{Y}^{*}(y) &:= \sum_{x} \tilde{q}_{X, Y}(x, y). \label{eq:LP_update_q_Y}
\end{align}
\item  For a fixed $q_{Y}$, $F_{\alpha}^{\text{LP}}(\tilde{q}_{X, Y}, q_{Y})$ is minimized by 
\begin{align}
\tilde{q}_{X, Y}^{*}(x, y) &= \bar{q}_{X}(x) \bar{q}_{Y\mid X}(y\mid x), \label{eq:LP_update_tilde_q_X_Y}
\end{align}
where 
\begin{align}
&\bar{q}_{X}(x) := \frac{p_{X}(x) \left( \sum_{y} p_{Y\mid X}(y\mid x)^{\alpha}q_{Y}(y)^{1-\alpha}\right)^{\frac{1}{\alpha}}}{\sum_{x}p_{X}(x) \left( \sum_{y} p_{Y\mid X}(y\mid x)^{\alpha}q_{Y}(y)^{1-\alpha}\right)^{\frac{1}{\alpha}}}, \label{eq:LP_update_vc_bar_q_X} \\ 
&\bar{q}_{Y\mid X}(y\mid x) := \frac{p_{Y\mid X}(y\mid x)^{\alpha}q_{Y}(y)^{1-\alpha}}{\sum_{y}p_{Y\mid X}(y\mid x)^{\alpha}q_{Y}(y)^{1-\alpha}}. \label{eq:LP_update_vc_bar_q_Y_X}
\end{align}
\end{enumerate}
For $\alpha \in (1, \infty)$, the following holds:
\begin{enumerate}
\item[3)] For a fixed $\tilde{q}_{X, Y}$, $\tilde{F}_{\alpha}^{\text{LP}}(\tilde{q}_{X, Y}, r_{X\mid Y})$ is maximized by 
\begin{align}
r_{X\mid Y}^{*}(x\mid y) &= \frac{\tilde{q}_{X, Y}(x, y)}{\sum_{x}\tilde{q}_{X, Y}(x, y)}. \label{eq:LP_update_vc_r_X_Y}
\end{align}
\item[4)]  For a fixed $r_{X\mid Y}$, $\tilde{F}_{\alpha}^{\text{LP}}(\tilde{q}_{X, Y}, r_{X\mid Y})$ is maximized by 
\begin{align}
q_{X, Y}^{*}(x, y) &:= \hat{q}_{X}(x)\hat{q}_{Y\mid X}(y\mid x), \label{eq:LP_update_vc_tilde_q_X_Y}
\end{align}
where 
\begin{align}
&\hat{q}_{X}(x) \notag \\ 
&:= \frac{p_{X}(x)^{\frac{\alpha}{2\alpha-1}}\left( \sum_{y}p_{Y\mid X}(y | x)r_{X\mid Y}(x|y)^{1-\frac{1}{\alpha}} \right)^{\frac{\alpha}{2\alpha-1}}}{\sum_{x}p_{X}(x)^{\frac{\alpha}{2\alpha-1}}\left(\sum_{y}p_{Y\mid X}(y|x)r_{X\mid Y}(x|y)^{1-\frac{1}{\alpha}} \right)^{\frac{\alpha}{2\alpha-1}}}, \label{eq:LP_update_vc_hat_q_X} \\ 
&\hat{q}_{Y\mid X}(y\mid x):= \frac{p_{Y\mid X}(y\mid x)r_{X\mid Y}(x\mid y)^{1-\frac{1}{\alpha}}}{\sum_{x}p_{Y\mid X}(y\mid x)r_{X\mid Y}(x\mid y)^{1-\frac{1}{\alpha}}}. \label{eq:LP_update_vc_hat_q_Y_X}
\end{align}
\end{enumerate}
\end{prop}
\begin{proof}
See Appendix \ref{proof:update_formulae_LP_MI_vc}.
\end{proof}

\subsection{AO Algorithms for Computing $\alpha$-Capacity} \label{sec:AO_alpha_capacity}

Arimoto \cite{arimoto1977} proposed an AO algorithm for computing the Sibson capacity based on the characterization in \eqref{eq:vc_Sibson_MI_01}. 
Subsequently, Arimoto \cite{BN01990060en} and Kamatsuka \textit{et al.} \cite{10619200} proposed similar AO algorithms, building on the characterizations of Arimoto and Sibson MI in Remark \ref{rem:vc_Arimoto} and \eqref{eq:vc_Sibson_MI_02}  (see \cite[Thm 1 and Table I]{10619200}). 
Kamatsuka \textit{et al.} \cite{kamatsuka2024_ISITA} derived an AO algorithm for computing the Augustin--Csisz{\' a}r capacity for $\alpha \in (1, \infty)$ based on the characterization in \eqref{eq:vc_AC_MI_02_03}. 
In \cite{kamatsuka2024_ISITA}, they also proposed another AO algorithm for the Sibson capacity by modifying the AO algorithm for computing the correct decoding exponent proposed by Jitsumatsu and Oohama \cite[Algorithm 1]{8889422}. 

In addition to these algorithms, an AO algorithm for $\alpha \in (1, \infty)$ can be derived from the characterizations in \eqref{eq:vc_Sibson_MI_03} and its proof with the following updating formulae:
\begin{prop} \label{prop:update_formulae_Sibson_03}
Let $\alpha \in (1, \infty)$ and $\tilde{F}_{\alpha}^{\text{S3}}(p_{X}, \tilde{q}_{X, Y}, r_{X\mid Y})$ be a functional defined in \eqref{eq:def_vc_Sibson_MI_03}. Then, 
\begin{enumerate}
\item For a fixed $(p_{X}, \tilde{q}_{X, Y})$, $\tilde{F}_{\alpha}^{\text{S3}}(p_{X}, \tilde{q}_{X, Y}, r_{X\mid Y})$ is maximized by 
\begin{align}
r_{X\mid Y}^{*}(x \mid y) &:= \frac{\tilde{q}_{X, Y}(x, y)}{\sum_{x}\tilde{q}_{X, Y}(x, y)}.
\end{align}
\item For a fixed $(\tilde{q}_{X, Y}, r_{X\mid Y})$, $\tilde{F}_{\alpha}^{\text{S3}}(p_{X}, \tilde{q}_{X, Y}, r_{X\mid Y})$ is maximized by 
\begin{align}
p_{X}^{*}(x) := \sum_{y}\tilde{q}_{X, Y}(x, y).
\end{align}
\item For a fixed $(p_{X}, r_{X\mid Y})$, $\tilde{F}_{\alpha}^{\text{S3}}(p_{X}, \tilde{q}_{X, Y}, r_{X\mid Y})$ is maximized by 
\begin{align}
\tilde{q}^{*}_{X, Y}(x, y) &:= \frac{p_{X}(x)^{\frac{1}{\alpha}}p_{Y\mid X}(y\mid x)r_{X\mid Y}(x\mid y)^{1-\frac{1}{\alpha}}}{\sum_{x,y}p_{X}(x)^{\frac{1}{\alpha}}p_{Y\mid X}(y\mid x)r_{X\mid Y}(x\mid y)^{1-\frac{1}{\alpha}}}.
\end{align}
\end{enumerate}
\end{prop}

\begin{remark}
The AO algorithm with the above updating formulae is equivalent to the AO algorithm for computing the correct decoding exponent without a cost constraint in \cite[Algorithm 2]{8889422}. 
\end{remark}

Similarly, a novel AO algorithm for computing the Lapidoth--Pfister capacity for $\alpha \in (1, \infty)$ can be derived from the characterization of Lapidoth--Pfister MI in \eqref{eq:vc_LP_MI_02_03} and its proof using the following updating formulae:
\begin{prop} \label{prop:update_formulae_LP_03}
Let $\alpha \in (1, \infty)$ and $\tilde{F}_{\alpha}^{\text{LP}}(p_{X}, \tilde{q}_{X, Y}, r_{X\mid Y})$ be a functional defined in \eqref{eq:def_vc_LP_MI_03}\footnote{Here, we have added $p_{X}$ as an argument for the functional $\tilde{F}_{\alpha}^{\text{LP}}$.}. Then, 
\begin{enumerate}
\item For a fixed $(p_{X}, \tilde{q}_{X, Y})$, $\tilde{F}_{\alpha}^{\text{LP}}(p_{X}, \tilde{q}_{X, Y}, r_{X\mid Y})$ is maximized by 
\begin{align}
r_{X\mid Y}^{*}(x \mid y) &:= \frac{\tilde{q}_{X, Y}(x, y)}{\sum_{x}\tilde{q}_{X, Y}(x, y)}.
\end{align}
\item For a fixed $(\tilde{q}_{X, Y}, r_{X\mid Y})$, $\tilde{F}_{\alpha}^{\text{LP}}(p_{X}, \tilde{q}_{X, Y}, r_{X\mid Y})$ is maximized by 
\begin{align}
p_{X}^{*}(x) := \sum_{y}\tilde{q}_{X, Y}(x, y). \label{eq:update_p_X_LP}
\end{align}
\item For a fixed $(p_{X}, r_{X\mid Y})$, $\tilde{F}_{\alpha}^{\text{LP}}(p_{X}, \tilde{q}_{X, Y}, r_{X\mid Y})$ is maximized by 
\begin{align}
\tilde{q}^{*}_{X, Y}(x, y) &:= \hat{q}_{X}(x)\hat{q}_{Y\mid X}(y\mid x), 
\end{align}
where $\hat{q}_{X}$ and $\hat{q}_{Y\mid X}$ are defined in \eqref{eq:LP_update_vc_hat_q_X} and \eqref{eq:LP_update_vc_hat_q_Y_X}, respectively.
\end{enumerate}
\end{prop}
\begin{proof}
This follows from Proposition \ref{prop:vc_Shannon_MI} and the proof of Proposition \ref{prop:update_formulae_LP_MI_vc}.
\end{proof}

\subsection{Global Convergence Properties of AO Algorithms} \label{ssec:global_convergence}
This subsection discusses the global convergence properties of each AO algorithm.  

\subsubsection{AO algorithms for computing $\alpha$-capacity}
Arimoto \cite[Thm 3]{1055640} and Kamatsuka \textit{et al.} \cite[Cor 2]{10619200} demonstrated the global convergence properties of the AO algorithms based on the characterization of the Sibson MI in \eqref{eq:vc_Sibson_MI_01} and \eqref{eq:vc_Sibson_MI_02} as well as the corresponding AO algorithms derived from the characterization of the Arimoto MI.
The modified Jitsuamtsu--Oohama algorithm proposed in \cite{kamatsuka2024_ISITA} also has this property (see \cite[Thm 2]{8889422}), while 
the property of the AO algorithm based on the characterization in \eqref{eq:vc_AC_MI_02_03} for computing the Augustin--Csisz{\' a}r capacity for $\alpha\in (1, \infty)$ remains undemonstrated.

Here, we prove that the AO algorithm for computing the Lapidoth--Pfister capacity for $\alpha \in (1, \infty)$ in Proposition \ref{prop:update_formulae_LP_03} has the global convergence property. 
First, the following lemma follows directly from Theorem \ref{thm:variational_characterization_LP_MI} and Proposition \ref{prop:update_formulae_LP_03}. 
\begin{lemma} \label{lem:LP_capacity_monotonically_incresing}
Let $\alpha \in (1, \infty)$ and $\{p_{X}^{(k)}\}_{k=0}^{\infty}$, $\{q_{X, Y}^{(k)}\}_{k=0}^{\infty}$, and $\{r_{X\mid Y}^{(k)}\}_{k=0}^{\infty}$ be sequences of distributions at the $k$th iteration of the AO algorithm. Then, the following holds:
\begin{align} 
&\tilde{F}_{\alpha}^{\text{LP}}(p_{X}^{(0)}, \tilde{q}_{X, Y}^{(0)}, r_{X\mid Y}^{(0)}) \leq \cdots  
\leq \tilde{F}_{\alpha}^{\text{LP}}(p_{X}^{(k)}, \tilde{q}_{X, Y}^{(k)}, r_{X\mid Y}^{(k)}) \notag \\ 
&\leq \tilde{F}_{\alpha}^{\text{LP}}(p_{X}^{(k+1)}, \tilde{q}_{X, Y}^{(k+1)}, r_{X\mid Y}^{(k)}) \leq \cdots \leq C_{\alpha}^{\text{LP}}.
\end{align}
\end{lemma}

Next, the approximation error at the $k$th iteration, $C_{\alpha}^{\text{LP}} - \tilde{F}_{\alpha}^{\text{C}}(p_{X}^{(k+1)}, \tilde{q}_{X, Y}^{(k+1)}, r_{X\mid Y}^{(k)})$, can be upper bounded as follows:
\begin{lemma}\label{lem:LP_capacity_upper_bound} 
Let $\alpha\in (1, \infty)$. Then, the approximation error has the following upper bound:
\begin{align}
&C_{\alpha}^{\text{LP}}- \tilde{F}_{\alpha}^{\text{LP}}(p_{X}^{(k+1)}, \tilde{q}_{X, Y}^{(k+1)}, r_{X\mid Y}^{(k)}) 
\leq \frac{\alpha}{\alpha-1} \sum_{x, y}\tilde{q}_{X, Y}^{\text{opt}}(x, y)  \notag \\ 
&\times  \log \frac{\tilde{q}_{X, Y}^{(k+1)}(x, y)}{\tilde{q}_{X, Y}^{(k)}(x, y)} + \sum_{x}\tilde{q}_{X}^{\text{opt}}(x)\log \frac{\tilde{q}_{X}^{(k+1)}(x)}{\tilde{q}_{X}^{(k)}(x)}, \label{eq:LP_capacity_upper_bound}
\end{align}
where $\tilde{q}_{X, Y}^{\text{opt}}$ be the optimal distributions that attain the Lapidoth--Pfister capacity $C_{\alpha}^{\text{LP}} = \max_{p_{X}} \max_{\tilde{q}_{X, Y}}\max_{r_{X\mid Y}} \tilde{F}_{\alpha}^{\text{LP}}(p_{X}, \tilde{q}_{X, Y}, r_{X\mid Y})$. 
\end{lemma}
\begin{proof}
See Appendix \ref{proof:LP_capacity_upper_bound}. 
\end{proof}

Summing \eqref{eq:LP_capacity_upper_bound} from $k=0$ to $k=N$ 
obtains the following upper bound that is independent of $N$:
\begin{align}
&\sum_{k=0}^{N} \left( C_{\alpha}^{\text{LP}} - \tilde{F}_{\alpha}^{\text{LP}}(p_{X}^{(k+1)}, \tilde{q}_{X, Y}^{(k+1)} r_{X\mid Y}^{(k)})  \right) \\ 
&\leq \frac{\alpha}{\alpha-1}\sum_{x,y} \tilde{q}_{X, Y}^{\text{opt}}(x, y)\log \frac{\tilde{q}_{X, Y}^{(N+1)}(x, y)}{\tilde{q}_{X, Y}^{(0)}(x, y)} \notag \\
&\qquad \qquad \qquad +\sum_{x}\tilde{q}_{X}^{\text{opt}}(x)\log \frac{\tilde{q}_{X}^{(N+1)}(x)}{\tilde{q}_{X}^{(0)}(x)} \\ 
&= \frac{\alpha}{\alpha-1} \left( D(\tilde{q}_{X, Y}^{\text{opt}} || \tilde{q}_{X, Y}^{(0)}) - D(\tilde{q}_{X, Y}^{\text{opt}} || \tilde{q}_{X, Y}^{(N+1)}) \right) \notag \\
&\qquad \qquad \qquad + D(\tilde{q}_{X}^{\text{opt}} || \tilde{q}_{X}^{(0)}) - D(\tilde{q}_{X}^{\text{opt}} || \tilde{q}_{X}^{(N+1)}) \\ 
&\leq \frac{\alpha}{\alpha-1} D(\tilde{q}_{X, Y}^{\text{opt}} || \tilde{q}_{X, Y}^{(0)}) + D(\tilde{q}_{X}^{\text{opt}} || \tilde{q}_{X}^{(0)}).
\end{align}
If we combine this upper bound with Lemma \ref{lem:LP_capacity_monotonically_incresing}, then the infinite series 
$\sum_{k=0}^{\infty} \left( C_{\alpha}^{\text{LP}}- \tilde{F}_{\alpha}^{\text{LP}}(p_{X}^{(k+1)}, \tilde{q}_{X, Y}^{(k+1)}, r_{X\mid Y}^{(k)}) \right)$ converges. 
Therefore, 
$C_{\alpha}^{\text{LP}} - \tilde{F}_{\alpha}^{\text{LP}}(p_{X}^{(k+1)}, \tilde{q}_{X, Y}^{(k+1)}, r_{X\mid Y}^{(k)})\to 0$ as $k\to \infty$, which implies the convergence result of the Lapidoth--Pfister capacity.
\begin{cor} \label{cor:AC_capacity_convergenece}
Let $\alpha\in (1, \infty)$. Then, 
\begin{align}
\lim_{k\to \infty} \tilde{F}_{\alpha}^{\text{LP}}(p_{X}^{(k+1)}, \tilde{q}_{X, Y}^{(k+1)}, r_{X\mid Y}^{(k)}) = C_{\alpha}^{\text{LP}}.
\end{align}
\end{cor}

\subsubsection{AO algorithm for computing Augustin--Csisz{\' a}r MI}
Augustin \cite[Lemma 34.2]{augusting_phd_thesis}, Karakos \textit{et al.} \cite[Thm 1]{4595361}, and Nakibo\v{g}ulu \cite[Lemma 13 (c)]{Nakiboglu:2019aa} proved that the AO algorithm based on \eqref{eq:vc_AC_MI_02_03} and \eqref{eq:def_vc_AC_MI_02} converges to the Augustin--Csisz{\' a}r MI for $\alpha \in (0, 1)$.
Similar to the proof of convergence of the Lapidoth--Pfister capacity, 
we show that the AO algorithm based on \eqref{eq:vc_AC_MI_02_03} and \eqref{eq:def_vc_AC_MI_03} globally converges when $\alpha \in (1, \infty)$.
Hence, it suffices to prove the following lemmas:
\begin{lemma} \label{lem:AC_monotonically_incresing}
Let $\alpha \in (1, \infty)$ and $\{\tilde{q}_{Y\mid X}^{(k)}\}_{k=0}^{\infty}$ and $\{r_{X\mid Y}^{(k)}\}_{k=0}^{\infty}$ be sequences of distributions at the $k$th iteration of the AO algorithm. Then, the following holds:
\begin{align} 
&\tilde{F}_{\alpha}^{\text{C}}(\tilde{q}_{Y\mid X}^{(0)}, r_{X\mid Y}^{(0)}) \leq \tilde{F}_{\alpha}^{\text{C}}(\tilde{q}_{Y\mid X}^{(1)}, r_{X\mid Y}^{(0)}) \leq \cdots \notag \\ 
&\leq \tilde{F}_{\alpha}^{\text{C}}(\tilde{q}_{Y\mid X}^{(k)}, r_{X\mid Y}^{(k)}) \leq \tilde{F}_{\alpha}^{\text{C}}(\tilde{q}_{Y\mid X}^{(k+1)}, r_{X\mid Y}^{(k)}) \leq \cdots \leq I_{\alpha}^{\text{C}}(X; Y). 
\end{align}
\end{lemma}
\begin{proof}
This follows directly from Theorem \ref{thm:variational_characterization_Csiszar_MI} and Proposition \ref{prop:update_formulae_Csiszar_MI}. 
\end{proof}

\begin{lemma}\label{lem:AC_upper_bound} 
Let $\alpha\in (1, \infty)$. Then, the approximation error has the following upper bound:
\begin{align}
&I_{\alpha}^{\text{C}}(X; Y)- \tilde{F}_{\alpha}^{\text{C}}(\tilde{q}_{Y\mid X}^{(k+1)}, r_{X\mid Y}^{(k)}) \notag \\ 
&\leq \frac{\alpha}{\alpha-1} \sum_{x, y} p_{X}(x)\tilde{q}_{Y\mid X}^{\text{opt}}(y\mid x)\log \frac{\tilde{q}_{Y\mid X}^{(k+1)}(y\mid x)}{\tilde{q}_{Y\mid X}^{(k)}(y\mid x)}, \label{eq:upper_bound}
\end{align}
where $\tilde{q}_{Y\mid X}^{\text{opt}} = \{\tilde{q}_{Y\mid X}^{\text{opt}}(\cdot | x)\}_{x\in \mathcal{X}}$ is the optimal family of distributions that achieves the maximum in \eqref{eq:vc_AC_MI_02_03} 
\end{lemma}
\begin{proof}
See Appendix \ref{proof:AC_upper_bound}.
\end{proof}

The following convergence result can then be obtained: 
\begin{cor} \label{cor:AC_MI_convergence}
Let $\alpha\in (1, \infty)$. Then, 
\begin{align}
\lim_{k\to \infty} \tilde{F}_{\alpha}^{\text{C}}(\tilde{q}_{Y\mid X}^{(k+1)}, r_{X\mid Y}^{(k)})
&= I_{\alpha}^{\text{C}}(X; Y).
\end{align}
\end{cor}

\subsubsection{AO algorithms for computing Lapidoth--Pfister MI}

The global convergence property of the AO algorithm based on Corollary \ref{cor:update_formulae_LP_MI} follows immediately from the sufficient condition for the global convergence of general AO algorithms \cite[Thm 10.5]{10.5555/1199866}.
\begin{prop} \label{prop:LP_MI_convergence}
Let $\alpha\in [1/2, 1) \cup (1, \infty)$ and $\{q_{X}^{(k)}\}_{k=0}^{\infty}$ and $\{q_{Y}^{(k)})\}_{k=0}^{\infty}$ be sequences of distributions at the $k$th iteration of the AO algorithm. Then, 
\begin{align}
\lim_{k\to \infty} D_{\alpha}(p_{X}p_{Y\mid X} || q_{X}^{(k)}q_{Y}^{(k)}) &= I_{\alpha}^{\text{LP}}(X; Y). \label{eq:convergence_LP}
\end{align}
\end{prop}
\begin{proof}
We can prove \eqref{eq:convergence_LP} by combining the sufficient condition from \cite[Thm 10.5]{10.5555/1199866} with the fact that $(q_{X}, q_{Y}) \mapsto D_{\alpha}(p_{X}p_{Y\mid X} || q_{X}q_{Y})$ is jointly convex for $[1/2, 1) \cup (1, \infty)$ \cite[Lemma 15]{e21080778}.
\end{proof}
\begin{remark}
Recently, Tsai \textit{et al.} \cite{tsai2024linearconvergencehilbertsprojective} provided a non-asymptotic analysis of the linear convergence rate of the AO algorithm.
\end{remark}

The global convergence property of the AO algorithm based on \eqref{eq:vc_LP_MI_02_03} for $\alpha \in (1, \infty)$ directly follows from Corollary \ref{cor:AC_capacity_convergenece} by replacing $p_{X}^{(k)}$ with $p_{X}$ for all $k\in \bn$.
For $\alpha \in (0, 1)$, it can be shown similarly to Corollaries \ref{cor:AC_capacity_convergenece} and \ref{cor:AC_MI_convergence}; 
thus, it suffices to show the following lemmas:
\begin{lemma} \label{lem:LP_monotonically_decresing}
Let $\alpha \in (0, 1)$ and $\{\tilde{q}_{X, Y}^{(k)}\}_{k=0}^{\infty}$ and $\{q_{Y}^{(k)}\}_{k=0}^{\infty}$ be sequences of distributions at the $k$th iteration of the AO algorithm. Then, the following holds:
\begin{align} 
&F_{\alpha}^{\text{LP}}(\tilde{q}_{X, Y}^{(0)}, q_{Y}^{(0)}) \geq F_{\alpha}^{\text{LP}}(\tilde{q}_{X, Y}^{(1)}, q_{Y}^{(0)}) \geq \cdots \notag \\ 
&\geq F_{\alpha}^{\text{LP}}(\tilde{q}_{X, Y}^{(k)}, q_{Y}^{(k)}) \geq F_{\alpha}^{\text{LP}}(\tilde{q}_{X, Y}^{(k+1)}, q_{Y}^{(k)}) \geq \cdots \geq I_{\alpha}^{\text{LP}}(X; Y). 
\end{align}
\end{lemma}
\begin{proof}
This follows directly from Theorem \ref{thm:variational_characterization_Csiszar_MI} and Proposition \ref{prop:update_formulae_Csiszar_MI}. 
\end{proof}

\begin{lemma}\label{lem:LP_upper_bound} 
Let $\alpha\in [1/2, 1)$. Then, the approximation error has the following upper bound:
\begin{align}
&F_{\alpha}^{\text{LP}}(\tilde{q}_{X, Y}^{(k+1)}, q_{Y}^{(k)}) - I_{\alpha}^{\text{LP}}(X; Y) \notag \\ 
&\qquad \leq \frac{\alpha}{1-\alpha} \sum_{x}\tilde{q}_{X}^{\text{opt}}(x)\log \frac{\tilde{q}_{X}^{(k+1)}(x)}{\tilde{q}_{X}^{(k)}(x)}, \label{eq:LP_upper_bound}
\end{align}
where $\tilde{q}_{X}^{\text{opt}}(x) := \sum_{y}\tilde{q}_{X, Y}^{\text{opt}}(x, y)$ and $\tilde{q}_{X, Y}^{\text{opt}}$ are the optimal joint distributions that achieve the minimum in \eqref{eq:vc_LP_MI_02_03}.
\end{lemma}
\begin{proof}
See Appendix \ref{proof:LP_upper_bound}.
\end{proof}

\begin{cor} \label{cor:LP_MI_convergence}
Let $\alpha\in [1/2, 1) \cup (1, \infty)$. Then, 
\begin{align}
\lim_{k\to \infty} F_{\alpha}^{\text{LP}}(\tilde{q}_{X, Y}^{(k+1)}, q_{Y}^{(k)}) &= I_{\alpha}^{\text{LP}}(X; Y), \quad \alpha \in [1/2, 1), \\
\lim_{k\to \infty} \tilde{F}_{\alpha}^{\text{LP}}(\tilde{q}_{X, Y}^{(k+1)}, r_{X\mid Y}^{(k)})&= I_{\alpha}^{\text{LP}}(X; Y),  \quad \alpha \in (1, \infty).
\end{align}
\end{cor}

\section{Comparison of AO algorithms for $\alpha$-Capacity} \label{ssec:comparison}
It is known that the Sibson, Arimoto, and Augustin--Csisz{\' a}r capacity are equivalent: $C_{\alpha}^{\text{S}} = C_{\alpha}^{\text{A}} = C_{\alpha}^{\text{C}}$ for $\alpha\in (0, 1)\cup (1, \infty)$ \cite{arimoto1977},\cite{370121},\cite{e22050526}. 
Moreover, the Lapidoth--Pfister capacity $C_{\alpha}^{\text{LP}}$ is equivalent to these capacities for $\alpha \in (1, \infty)$ \cite[Thm 4]{e22050526}, \cite[p.4]{e21100969}. 
Thus, all AO algorithms obtained so far compute equivalent quantities.
This section compares their performance through numerical examples.
From this point onward, we refer to an AO algorithm based on the variational characterization using a functional $F_{\alpha}^{\text{X}}(\cdot)$ or $\tilde{F}_{\alpha}^{\text{X}}(\cdot)$ as Algorithm \text{X}. 
Moreover, we refer to the modified AO algorithm of \cite[Algorithm 1]{8889422} as Algorithm JO.

Kamatsuka \textit{et al.} \cite[Cor 1]{10619200} proved that the AO algorithms S1, S2, and the algorithms based on the variational characterizations of the Arimoto MI are all equivalent under certain conditions on initial distributions. 
Similarly to \cite[Lemma 7 and 8]{8889422}, it can be shown that the Algorithm S1 achieves a greater increase in the objective function per update than Algorithm S3.
Thus, we numerically compared the AO algorithms S1, JO, C, and LP for $\alpha \in (1, \infty)$. 

Define a channel 
\begin{align}
p_{Y\mid X} &:= 
\begin{bmatrix}
0.259 & 0.463 & 0.278 \\ 
0.328 & 0.172 & 0.500 \\ 
0.425 & 0.225 & 0.350
\end{bmatrix}, 
\end{align}
where $(i, j)$-element of the channel matrix\footnote{The elements of the channel matrix $p_{Y\mid X}$ are randomly generated from a uniform distribution. Additional numerical examples were computed using the AO algorithms for other randomly generated channel matrices that did not appear in this paper, wherein the behavior of the algorithms is generally the same.} 
corresponds to the conditional probability $p_{Y\mid X}(j | i), i\in \mathcal{X}=\{1, 2, 3\}, j\in \mathcal{Y}=\{1, 2, 3\}$.
Table \ref{tab:alpha_capacity} shows the number of iterations $N$ and approximate values of $\alpha$-capacity  of the channel $p_{Y\mid X}$, denoted as $F^{(N)}$,  for $\alpha\in \{1.03, 1.5, 2.0, 5.0\}$ computed by each algorithm with specific initial conditions.
Here, $N$ represents the number of iterations needed to satisfy the stopping conditions. The distributions $u_{X}$ and $u_{X,Y}$ are uniform distributions on $\mathcal{X}$ and $\mathcal{X}\times \mathcal{Y}$, respectively.
Figure \ref{fig:numerical_example} illustrates transitions of the approximate values of $\alpha$-capacity $F^{(k)}$ as $k$ increases, computed by each algorithm with specific initial conditions.
Table \ref{tab:alpha_capacity} and Figure \ref{fig:numerical_example} confirm that Algorithm S1 has the fastest convergence speed.

\begin{table*}[h]
  \caption{Approximate values of $\alpha$-capacity and the number of iterations $(F^{(N)}, N)$}
  \label{tab:alpha_capacity}
  \centering
  \resizebox{1.\textwidth}{!}{
  \begin{tabular}{@{} lllll@{}}
    \toprule
    \thead{Algorithm and initial distribution}        & \thead{$\alpha=1.03$} & \thead{$\alpha=1.5$} & \thead{$\alpha=2.0$} & \thead{$\alpha=5.0$} \\ 
    \midrule
    Algorithm S1 with $p_{X}^{(0)} = u_{X}$
    & $(0.054204678, \textbf{851})$ & $(0.07617995, \textbf{738})$ & $(0.097030615, \textbf{790})$ & $(0.183426237, \textbf{20})$ \\ 
	Algorithm JO with $q_{X, Y}^{(0)}= u_{X, Y}$ 
	& $(0.054201694, 19076)$ & $(0.07617965, 2744)$ & $(0.097030139,2258)$ & $(0.183426230,113)$ \\
	Algorithm JO with $q_{X, Y}^{(0)}= u_{X} p_{Y\mid X}$ 
	& $(0.054201694, 19075)$ & $(0.07617965, 2743)$ & $(0.097030139, 2257)$ & $(0.183426230, 112)$ \\ 
	Algorithm C with $p_{X}^{(0)}\tilde{q}_{Y\mid X}^{(0)}= u_{X, Y}$ 
	& $(0.054204678, 912)$ & $(0.07617991, 1059)$ & $(0.097030445, 1347)$ & $(0.183426230, 97)$ \\ 
	Algorithm C with $p_{X}^{(0)}\tilde{q}_{Y\mid X}^{(0)} = u_{X} p_{Y\mid X}$ 
	& $(0.054204678, 874)$ & $(0.07617991, 1055)$ & $(0.097030445, 1344)$ & $(0.183426231, 95)$ \\ 
	Algorithm LP with $\tilde{q}_{X, Y}^{(0)}= u_{X, Y}$ 
	& $(0.054204763, 19540)$ & $(0.07618003, 3512)$ & $(0.097030811, 3044)$ & $(0.183426237, 177)$ \\ 
	Algorithm LP with $\tilde{q}_{X, Y}^{(0)} = u_{X} p_{Y\mid X}$ 
	& $(0.054204763, 19539)$ & $(0.07618003, 3511)$ & $(0.097030811, 3043)$ & $(0.183426237, 176)$ \\ 
    \bottomrule
  \end{tabular}
  }
\end{table*}

\begin{figure*}[t]
\centering
\begin{tabular}{@{}cc@{}}
\includegraphics[width=8cm, clip]{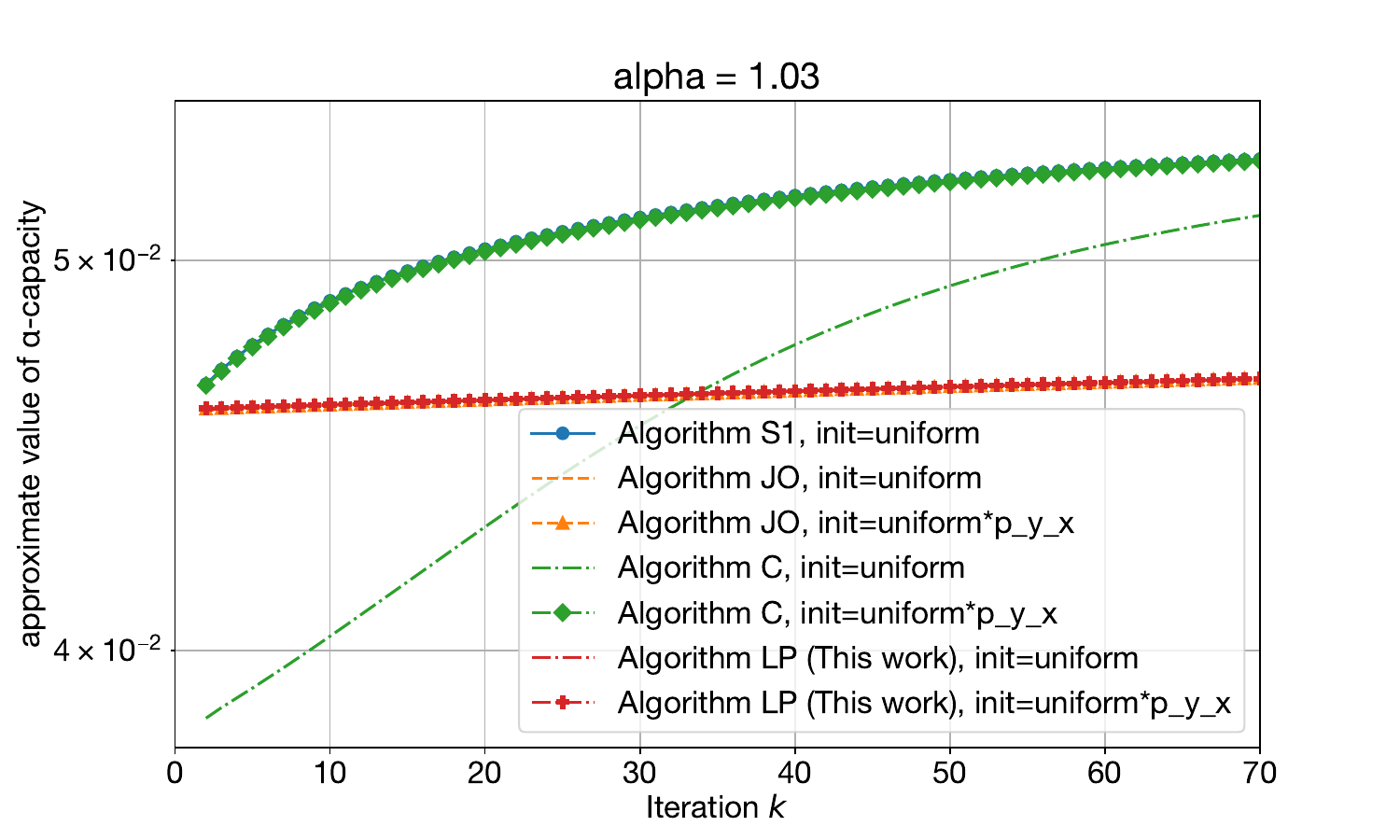} & \includegraphics[width=8cm, clip]{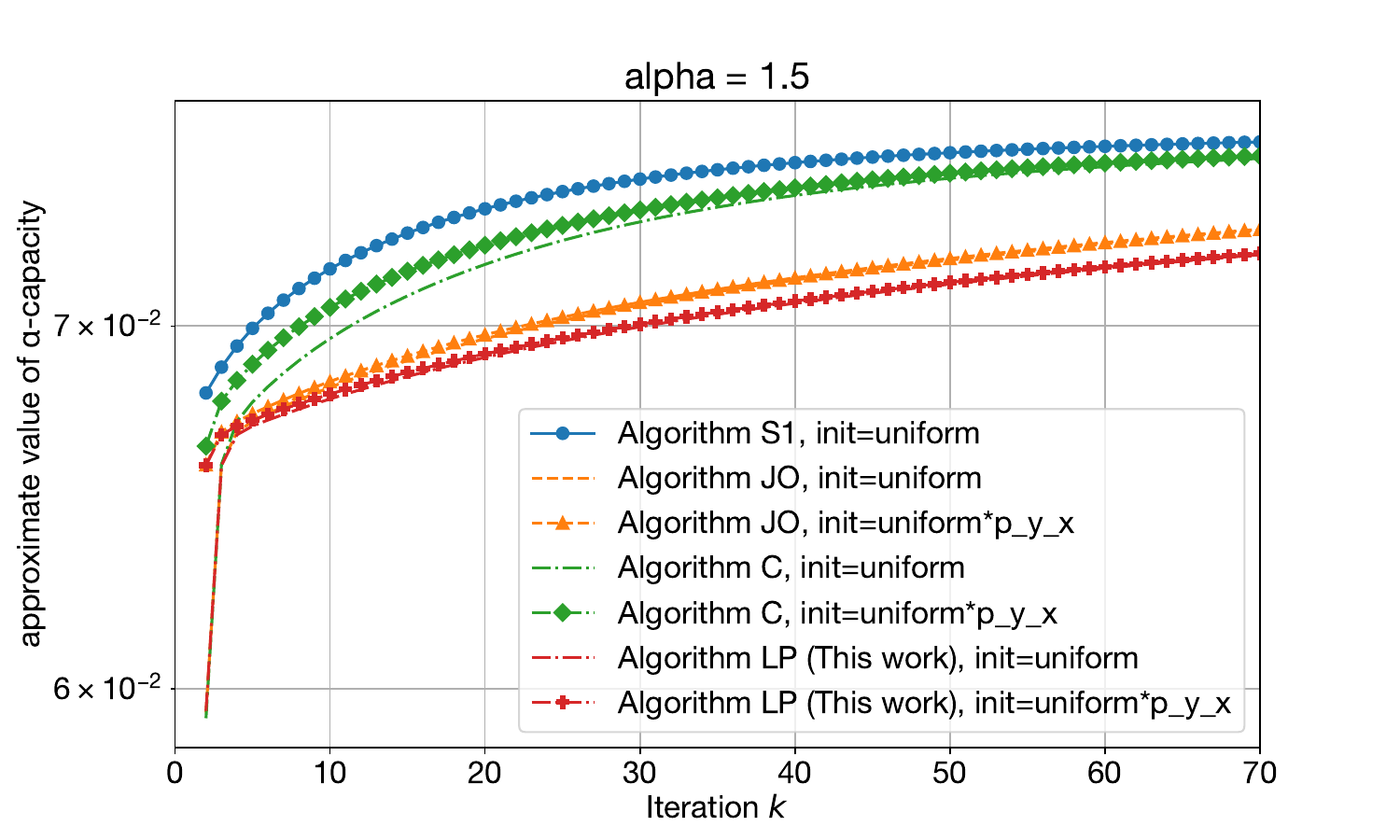} \\ 
(a) $\alpha=1.03$ & (b) $\alpha=1.5$ \\ 
\includegraphics[width=8cm, clip]{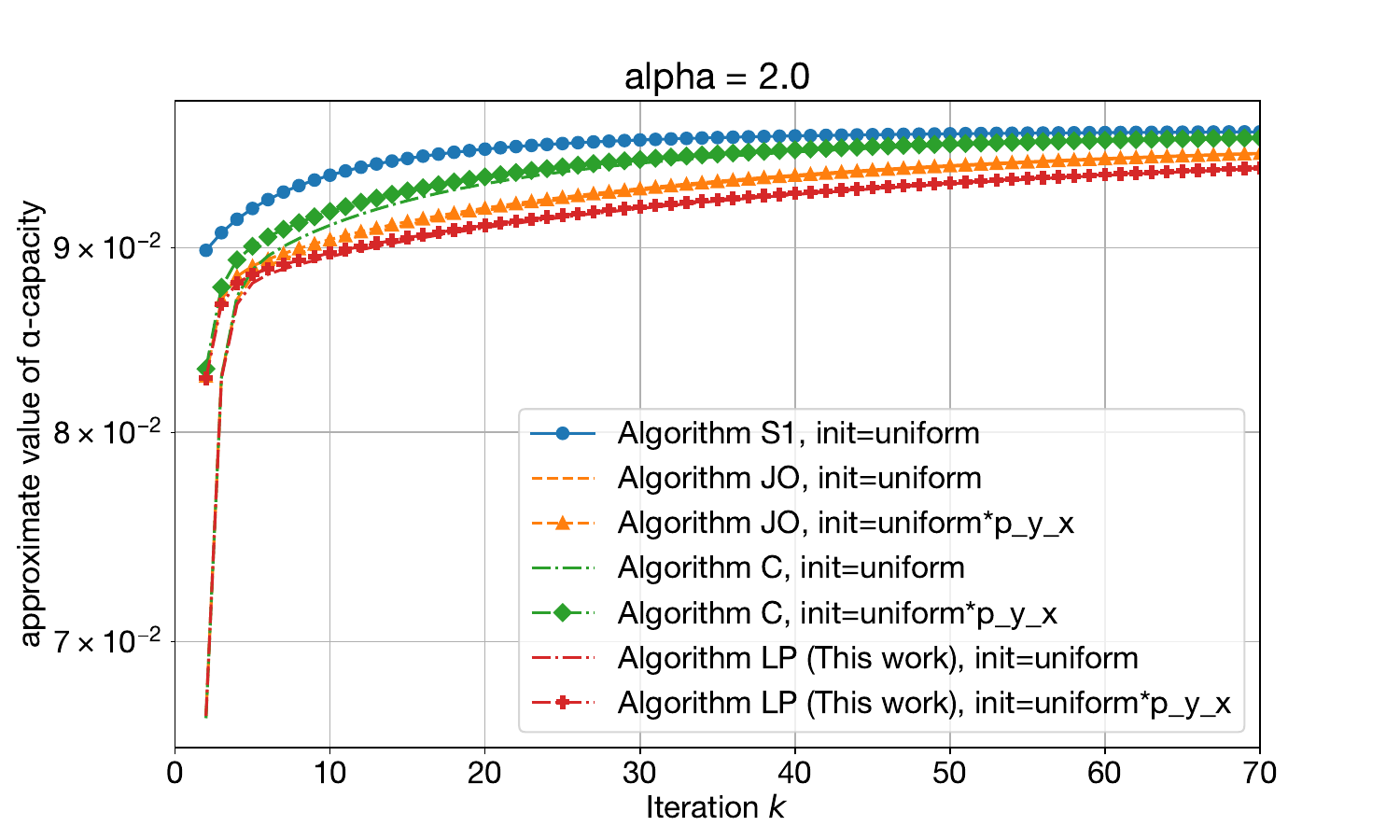} & \includegraphics[width=8cm, clip]{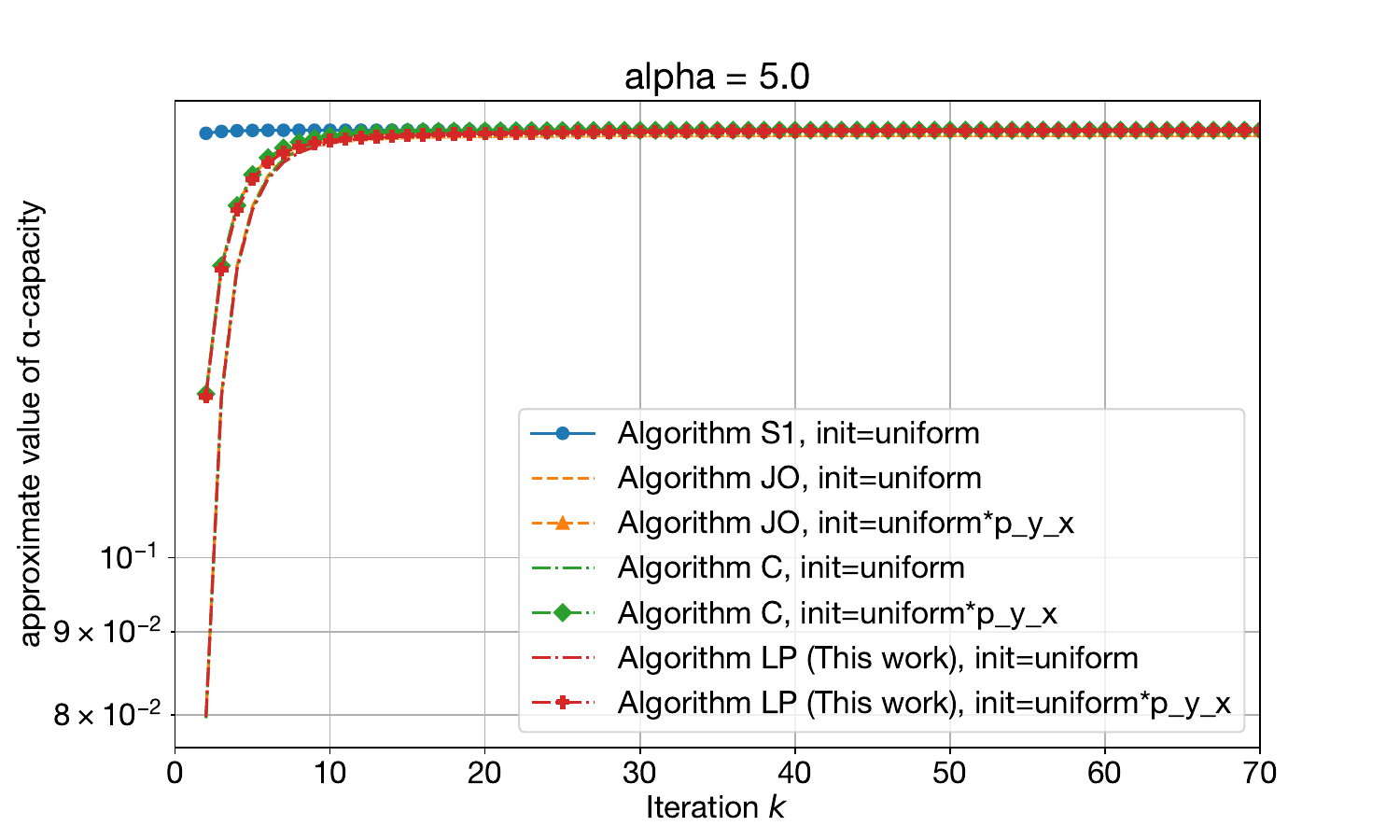} \\  
(c) $\alpha=2.0$  & (d) $\alpha=5.0$
\end{tabular}
\caption{Transitions of the approximate value of $\alpha$-capacity $F^{(k)}$ as $k$ increases for (a) $\alpha=1.03$, (b) $\alpha=1.5$, (c) $\alpha=2.0$, and (d) $\alpha=5.0$. The blue solid curve represents Algorithm S1, the orange dashed curve denotes Algorithm JO (modified Jitsumasu--Oohama algorithm), the green dash-dotted curve represents Algorithm C, and the red dash-dot curve indicates Algorithm LP.}
\label{fig:numerical_example}
\end{figure*}

\section{Conclusion}\label{sec:conclusion}
This paper provides variational characterizations of $\alpha$-MI.
Based on these characterizations, several novel AO algorithms for computing $\alpha$-MI and $\alpha$-capacity are presented. 
The numerical results show that the AO algorithm based on the Sibson MI's characterization has the fastest convergence speed.
Future work will focus on proving the nonasymptotic convergence analysis of these AO algorithms.

\section*{Acknowledgment}
The authors would like to thank Bar{\i}\c{s} Nakibo\v{g}lu and anonymous reviewers for their helpful comments.

\appendices
\section{Proof of Theorem \ref{thm:variational_characterization_Sibson_MI}} \label{proof:variational_characterization_Sibson_MI}
The characterizations in \eqref{eq:def_vc_Sibson_MI_01} and \eqref{eq:def_vc_Sibson_MI_02} were originally proved by Arimoto\cite[Thm 4]{arimoto1977} and Kamatsuka \textit{et al.} \cite[Thm 1]{10619200}, respectively. 
Here, we present a novel proof for \eqref{eq:def_vc_Sibson_MI_01} using the variational characterizations of the Shannon MI in Proposition \ref{prop:vc_Shannon_MI}. 

We define the distribution as follows:
\begin{align}
\hat{q}_{X, Y} (x, y)
&:= \frac{p_{X}(x)^{\frac{1}{\alpha}}p_{Y\mid X}(y\mid x)r_{X\mid Y}(x\mid y)^{1-\frac{1}{\alpha}}}{\sum_{x,y}p_{X}(x)^{\frac{1}{\alpha}}p_{Y\mid X}(y\mid x)r_{X\mid Y}(x\mid y)^{1-\frac{1}{\alpha}}}. \label{eq:hat_q_x_y}
\end{align}

Then, for $\alpha \in (0, 1)$, 
\begin{align}
&I_{\alpha}^{\text{S}}(X; Y) \overset{(a)}{=} \min_{\tilde{q}_{X, Y}} \Bigl\{ \frac{\alpha}{1-\alpha}D(\tilde{q}_{X, Y} || \tilde{q}_{X}p_{Y\mid X}) + I(\tilde{q}_{X}, \tilde{q}_{Y\mid X}) \notag \\ 
&\qquad \qquad \qquad \qquad + \frac{1}{1-\alpha}D(\tilde{q}_{X} || p_{X}) \Bigr\} \label{eq:S_proof_begin} \\ 
&\overset{(b)}{=} \min_{\tilde{q}_{X, Y}}\max_{r_{X\mid Y}} \Bigl\{\frac{\alpha}{1-\alpha}D(\tilde{q}_{X, Y} || \tilde{q}_{X}p_{Y\mid X})  \notag \\ 
&\qquad + \vE_{X, Y}^{\tilde{q}_{X, Y}}\left[\log \frac{r_{X\mid Y}(X\mid Y)}{\tilde{q}_{X}(X)}\right] + \frac{1}{1-\alpha}D(\tilde{q}_{X} || p_{X}) \Bigr\} \\ 
&\overset{(c)}{=} \max_{r_{X\mid Y}}\min_{\tilde{q}_{X, Y}} \Bigl\{\frac{\alpha}{1-\alpha}D(\tilde{q}_{X, Y} || \tilde{q}_{X}p_{Y\mid X})  \notag \\ 
&\qquad + \vE_{X, Y}^{\tilde{q}_{X, Y}}\left[\log \frac{r_{X\mid Y}(X\mid Y)}{\tilde{q}_{X}(X)}\right] + \frac{1}{1-\alpha}D(\tilde{q}_{X} || p_{X}) \Bigr\} \\ 
&=\max_{r_{X\mid Y}}\min_{\tilde{q}_{X, Y}} \Bigl\{D(\tilde{q}_{X,Y} || \hat{q}_{X, Y}) + F_{\alpha}^{\text{S1}}(p_{X}, r_{X\mid Y})\Bigr\} \\
&\overset{(d)}{=} \max_{r_{X\mid Y}} F_{\alpha}^{\text{S1}}(p_{X}, r_{X\mid Y}),  \label{eq:S_proof_end}
\end{align}
where $(a)$ follows from \eqref{eq:vc_Sibson_MI_KL}, 
$(b)$ follows from \eqref{eq:max_characterization_Shannon_MI}, 
$(c)$ follows from the minimax theorem (see, e.g., \cite[Thm 4.2]{pjm/1103040253})\footnote{Notably, $\tilde{F}_{\alpha}^{\text{S3}}(p_{X}, \tilde{q}_{X, Y}, r_{X\mid Y}):= \frac{\alpha}{1-\alpha}D(\tilde{q}_{X, Y} || \tilde{q}_{X}p_{Y\mid X}) + \vE_{X, Y}^{\tilde{q}_{X, Y}}\left[\log \frac{r_{X\mid Y}(X\mid Y)}{\tilde{q}_{X}(X)}\right] + \frac{1}{1-\alpha}D(\tilde{q}_{X} || p_{X}), \alpha \in (0, 1)$ is concave for $r_{X\mid Y}$ (see, e.g. \cite[Sec. 10.3.2]{10.5555/1199866}) and convex for $\tilde{q}_{X, Y}$, which can be easily verified.}, 
and $(d)$ follows from \eqref{eq:hat_q_x_y} and the properties of the Kullback--Leibler divergence (see, e.g., \cite[Thm 2.6.3]{Cover:2006:EIT:1146355}). 
For $\alpha \in (1, \infty)$, Theorem \ref{thm:variational_characterization_Sibson_MI} can be shown by replacing $\min_{\tilde{q}_{X, Y}}$ with $\max_{\tilde{q}_{X, Y}}$ in \eqref{eq:S_proof_begin}--\eqref{eq:S_proof_end}.

\section{Proof of Theorem \ref{thm:variational_characterization_Csiszar_MI}} \label{proof:variational_characterization_Csiszar_MI}
We define the distributions as follows.
\begin{align}
&\hat{q}_{Y\mid X}(y\mid x) := \frac{p_{Y\mid X}(y\mid x)r_{X\mid Y}(x\mid y)^{1-\frac{1}{\alpha}}}{\sum_{y}p_{Y\mid X}(y\mid x)r_{X\mid Y}(x\mid y)^{1-\frac{1}{\alpha}}} \label{eq:hat_q_Y_X} \\ 
&\hat{q}_{X}(x) := \frac{p_{X}(x)^{\frac{\alpha}{2\alpha-1}} \left( \sum_{y}p_{Y\mid X}(y | x)r_{X\mid Y}(x | y)^{1-\frac{1}{\alpha}} \right)^{\frac{\alpha}{2\alpha-1}}}{\sum_{x}p_{X}(x)^{\frac{\alpha}{2\alpha-1}} \left( \sum_{y}p_{Y\mid X}(y | x)r_{X\mid Y}(x | y)^{1-\frac{1}{\alpha}} \right)^{\frac{\alpha}{2\alpha-1}}}, \label{eq:update_LP_q_X}
\end{align}
Then, for $\alpha \in (0, 1)$, 
\begin{align}
&I_{\alpha}^{\text{C}}(X; Y) \notag \\ 
&\overset{(a)}{=} \min_{\tilde{q}_{Y\mid X}} \left\{I(p_{X}, \tilde{q}_{Y\mid X}) + \frac{\alpha}{1-\alpha} D(p_{X}\tilde{q}_{Y\mid X} || p_{X}p_{Y\mid X})\right\} \label{eq:C_proof_begin} \\ 
&\overset{(b)}{=} \min_{\tilde{q}_{Y\mid X}} \max_{r_{X\mid Y}} 
\Biggl\{\vE_{X, Y}^{p_{X}\tilde{q}_{Y\mid X}} \left[\log \frac{r_{X\mid Y}(X\mid Y)}{p_{X}(X)}\right] \notag \\ 
&\qquad \qquad \qquad \qquad + \frac{\alpha}{1-\alpha} D(p_{X}\tilde{q}_{Y\mid X} || p_{X}p_{Y\mid X})\Biggr\} \\ 
&\overset{(c)}{=} \max_{r_{X\mid Y}} \min_{\tilde{q}_{Y\mid X}}
\Biggl\{\vE_{X, Y}^{p_{X}\tilde{q}_{Y\mid X}} \left[\log \frac{r_{X\mid Y}(X\mid Y)}{p_{X}(X)}\right] \notag \\ 
&\qquad \qquad \qquad \qquad + \frac{\alpha}{1-\alpha} D(p_{X}\tilde{q}_{Y\mid X} || p_{X}p_{Y\mid X})\Biggr\} \\
&= \max_{r_{X\mid Y}} \min_{\tilde{q}_{Y\mid X}} \Biggl\{H(p_{X}) \notag \\\ 
&+\frac{\alpha}{1-\alpha} \vE_{X}^{p_{X}} \left[D\left(\tilde{q}_{Y\mid X}(\cdot \mid X) || \hat{q}_{Y\mid X}(\cdot \mid X)\right)\right] 
\notag \\
&+ \frac{\alpha}{\alpha-1} \sum_{x}p_{X}(x)\log \sum_{y} p_{Y\mid X}(y | x)r_{X\mid Y}(x | y)^{1-\frac{1}{\alpha}}\Biggr\} \label{eq:C_proof_end} \\
&\overset{(d)}{=} \eqref{eq:vc_AC_MI_01}, 
\end{align}
where $\hat{q}_{Y\mid X}(\cdot | x)$ is defined in \eqref{eq:hat_q_Y_X} and 
$(a)$ follows from \cite[Thm 1]{6034266}, 
$(b)$ follows from \eqref{eq:max_characterization_Shannon_MI}, 
$(c)$ follows from the minimax theorem\footnote{Notably, $\tilde{F}_{\alpha}^{\text{C}}(\tilde{q}_{Y\mid X}, r_{X\mid Y}):= \vE_{X, Y}^{p_{X}\tilde{q}_{Y\mid X}} \left[\log \frac{r_{X\mid Y}(X\mid Y)}{p_{X}(X)}\right] + \frac{\alpha}{1-\alpha} D(p_{X}\tilde{q}_{Y\mid X} || p_{X}p_{Y\mid X}), \alpha \in (0, 1)$ is concave for $r_{X\mid Y}$ and linear (hence, convex) for $\tilde{q}_{Y\mid X}$. },  
and $(d)$ follows from the properties of the Kullback--Leibler divergence (see, e.g., \cite[Thm 2.6.3]{Cover:2006:EIT:1146355}). 
Eq. \eqref{eq:vc_AC_MI_02_03} for $\alpha \in (0, 1)$ can be obtained from \eqref{eq:min_characterization_Shannon_MI} and \eqref{eq:C_proof_begin} by replacing $I(p_{X}, \tilde{q}_{Y\mid X})$ with $\min_{q_{Y}} D(p_{X}\tilde{q}_{Y\mid X} || p_{X}q_{Y})$, which was originally proved by Karakos \textit{et al.} \cite{4595361}.
For $\alpha \in (1, \infty)$, Theorem \ref{thm:variational_characterization_Csiszar_MI} can be shown by replacing $\min_{\tilde{q}_{Y\mid X}}$ with $\max_{\tilde{q}_{Y\mid X}}$ in Eqs. \eqref{eq:C_proof_begin}--\eqref{eq:C_proof_end}.

\section{Proof of Theorem \ref{thm:variational_characterization_LP_MI}} \label{proof:variational_characterization_LP_MI}
We define the distributions as follows.
\begin{align}
&\hat{q}_{Y\mid X}(x | y) := \frac{p_{Y\mid X}(y | x) r_{X\mid Y}(x | y)^{1-\frac{1}{\alpha}}}{\sum_{y}p_{Y\mid X}(y | x) r_{X\mid Y}(x | y)^{1-\frac{1}{\alpha}}}, \label{eq:update_LP_q_Y_given_X} \\ 
&\hat{q}_{X}(x) := \frac{p_{X}(x)^{\frac{\alpha}{2\alpha-1}} \left( \sum_{y}p_{Y\mid X}(y | x)r_{X\mid Y}(x | y)^{1-\frac{1}{\alpha}} \right)^{\frac{\alpha}{2\alpha-1}}}{\sum_{x}p_{X}(x)^{\frac{\alpha}{2\alpha-1}} \left( \sum_{y}p_{Y\mid X}(y | x)r_{X\mid Y}(x | y)^{1-\frac{1}{\alpha}} \right)^{\frac{\alpha}{2\alpha-1}}}. \label{eq:update_LP_q_X}
\end{align}
Then, for $\alpha \in (1, \infty)$, 
\begin{align}
&I_{\alpha}^{\text{LP}}(X; Y) 
\overset{(a)}{=} \max_{\tilde{q}_{X, Y}}\left\{ \frac{\alpha}{1-\alpha} D(\tilde{q}_{X, Y} || p_{X}p_{Y\mid X}) + I(\tilde{q}_{X}, \tilde{q}_{Y\mid X}) \right\} \label{eq:LP_proof_begin} \\
&\overset{(b)}{=} \max_{\tilde{q}_{X, Y}}\max_{r_{X\mid Y}} \Bigl\{ \frac{\alpha}{1-\alpha} D(\tilde{q}_{X, Y} || p_{X}p_{Y\mid X}) \notag \\
&\qquad \qquad \qquad + \vE_{X, Y}^{\tilde{q}_{X, Y}}\left[\log \frac{r_{X\mid Y}(X\mid Y)}{\tilde{q}_{X}(X)}\right] \Bigr\} \\ 
&= \max_{\tilde{q}_{X, Y}}\max_{r_{X\mid Y}} \Bigl\{ \frac{\alpha}{1-\alpha} D(\tilde{q}_{X, Y} || \tilde{q}_{X}p_{Y\mid X}) \notag \\ 
&\qquad + \vE_{X, Y}^{\tilde{q}_{X, Y}}\left[\log \frac{r_{X\mid Y}(X\mid Y)}{\tilde{q}_{X}(X)}\right] + D(\tilde{q}_{X} || p_{X})\Bigr\} \\
&= \max_{r_{X\mid Y}}\max_{\tilde{q}_{X}}\max_{\tilde{q}_{Y\mid X}} \Bigl\{\frac{\alpha}{1-\alpha} \vE_{X}^{\tilde{q}_{X}}\left[D(\tilde{q}_{Y\mid X}(\cdot \mid X) || \hat{q}_{Y\mid X}(\cdot \mid X))\right] \notag \\ 
& + H(\tilde{q}_{Y}) + \frac{\alpha}{1-\alpha}D(\tilde{q}_{X} || p_{X})\Bigr\} \\ 
&\overset{(c)}{=}\max_{r_{X\mid Y}}\max_{\tilde{q}_{X}} \Bigl\{\frac{2\alpha-1}{1-\alpha} D(\tilde{q}_{X} || \hat{q}_{X}) + \eqref{eq:vc_LP_MI_01}\Bigr\} \\ 
&\overset{(d)}{=} \eqref{eq:vc_LP_MI_01}, \label{eq:LP_proof_end}
\end{align}
where 
$(a)$ follows from \eqref{eq:vc_LP_MI_KL}, $(b)$ follows from \eqref{eq:max_characterization_Shannon_MI}, 
and $(c)$ and $(d)$ follow from the properties of the Kullback--Leibler divergence. 
For $\alpha \in (0, 1)$, Eq. \eqref{eq:vc_LP_MI_02_03} for $\alpha \in (0, 1)$ can be obtained by replacing $\max_{\tilde{q}_{X, Y}}$ with $\min_{\tilde{q}_{X, Y}}$ and $I(\tilde{q}_{X}, \tilde{q}_{Y\mid X})$ with $\min_{q_{Y}} D(\tilde{q}_{X}\tilde{q}_{Y\mid X} || \tilde{q}_{X}q_{Y})$ in \eqref{eq:LP_proof_begin}.

\section{Proof of Proposition \ref{prop:update_formulae_LP_MI_vc}}\label{proof:update_formulae_LP_MI_vc}
Eqs. \eqref{eq:LP_update_q_Y}, \eqref{eq:LP_update_vc_r_X_Y}, and \eqref{eq:LP_update_vc_tilde_q_X_Y} immediately follow from the proof of Theorem \ref{thm:variational_characterization_LP_MI} and Proposition \ref{prop:vc_Shannon_MI}. 
Here, we prove only \eqref{eq:LP_update_tilde_q_X_Y}. Let $\bar{q}_{X}$ and $\bar{q}_{Y\mid X}$ be distributions defined in \eqref{eq:LP_update_vc_bar_q_X} and \eqref{eq:LP_update_vc_bar_q_Y_X}, respectively. Let $\tilde{q}_{X}(x):=\sum_{y}\tilde{q}_{X, Y}(x, y)$.
Then,
\begin{align}
&F_{\alpha}^{\text{LP}}(\tilde{q}_{X, Y}, q_{Y}) = \frac{1}{1-\alpha} \vE_{X, Y}^{\tilde{q}_{X}} \left[D(\tilde{q}_{Y\mid X}(\cdot \mid X) || \bar{q}_{Y\mid X}(\cdot \mid X))\right] \notag \\ 
&+ \frac{1}{\alpha-1} \sum_{x}\tilde{q}_{X}(x)\log \sum_{y}p_{Y\mid X}(y\mid x)^{\alpha}q_{Y}(y)^{1-\alpha} \notag \\ 
&+ \frac{\alpha}{1-\alpha}D(\tilde{q}_{X} || p_{X}) \\ 
&\geq \frac{1}{\alpha-1} \sum_{x}\tilde{q}_{X}(x)\log \sum_{y}p_{Y\mid X}(y\mid x)^{\alpha}q_{Y}(y)^{1-\alpha} \notag \\ 
&+ \frac{\alpha}{1-\alpha}D(\tilde{q}_{X} || p_{X}) \label{eq:ineq_tilde_q_Y_X} \\ 
&= \frac{\alpha}{1-\alpha} D(\tilde{q}_{X} || \bar{q}_{X}) \notag \\ 
&+ \frac{\alpha}{\alpha-1}\log \sum_{x}p_{X}(x) \left( \sum_{y}p_{Y\mid X}(y\mid x)^{\alpha}q_{Y}(y)^{1-\alpha} \right)^{\frac{1}{\alpha}} \\ 
&\geq \frac{\alpha}{\alpha-1}\log \sum_{x}p_{X}(x) \left( \sum_{y}p_{Y\mid X}(y\mid x)^{\alpha}q_{Y}(y)^{1-\alpha} \right)^{\frac{1}{\alpha}}, \label{eq:ineq_tilde_q_X}
\end{align}
where the equality in \eqref{eq:ineq_tilde_q_Y_X} is achieved when $\tilde{q}_{Y\mid X} = \bar{q}_{Y\mid X}$; 
the equality in \eqref{eq:ineq_tilde_q_X} is achieved when $\tilde{q}_{X} = \bar{q}_{X}$. 

\section{Proof of Lemma \ref{lem:LP_capacity_upper_bound}} \label{proof:LP_capacity_upper_bound}
\begin{proof}
Let $\tilde{q}_{X, Y}^{\text{opt}}$ and $p_{X}^{\text{opt}}$ be the optimal distributions that attain the Lapidoth--Pfister capacity $C_{\alpha}^{\text{LP}} = \max_{p_{X}} \max_{\tilde{q}_{X, Y}}\max_{r_{X\mid Y}} \tilde{F}_{\alpha}^{\text{LP}}(p_{X}, \tilde{q}_{X, Y}, r_{X\mid Y})$. 
Let $\tilde{q}_{X, Y}^{(k)}(x, y)=\tilde{q}_{X}^{(k)}(x)\tilde{q}_{Y\mid X}^{(k)}(y|x), \tilde{q}_{Y}^{(k)}(y):=\sum_{x}\tilde{q}_{X, Y}^{(k)}(x, y)$ and $\tilde{q}_{X, Y}^{\text{opt}}(x, y) = \tilde{q}_{X}^{\text{opt}}(x)\tilde{q}_{Y\mid X}^{\text{opt}}(y|x), \tilde{q}_{Y}^{\text{opt}}(y):= \sum_{x}\tilde{q}_{X, Y}^{\text{opt}}(x, y)$.
Then, from Theorem \ref{thm:variational_characterization_LP_MI} and Proposition \ref{prop:update_formulae_LP_MI_vc}, we can obtain the following:
\begin{align}
&\tilde{F}_{\alpha}^{\text{LP}}(p_{X}^{(k+1)}, \tilde{q}_{X, Y}^{(k+1)}, r_{X\mid Y}^{(k)}) = \max_{\tilde{q}_{X, Y}} \tilde{F}_{\alpha}^{\text{LP}}(p_{X}^{(k+1)}, \tilde{q}_{X, Y}, r_{X\mid Y}^{(k)}) \\ 
&\overset{(a)}{=} \frac{2\alpha-1}{\alpha-1} \notag \\
&\times \log \sum_{x} p_{X}^{(k+1)}(x)^{\frac{\alpha}{2\alpha-1}} \left( \sum_{y} p_{Y\mid X}(y|x)r_{X\mid Y}(x|y)^{1-\frac{1}{\alpha}} \right)^{\frac{\alpha}{2\alpha-1}}, \label{eq:tilde_alpha_LP_capacity} \\
&\sum_{y}p_{Y\mid X}(y|x)r_{X\mid Y}^{(k)}(x|y)^{1-\frac{1}{\alpha}}
\overset{(b)}{=} \frac{p_{Y\mid X}(y|x)r_{X\mid Y}^{(k)}(x|y)^{1-\frac{1}{\alpha}}}{\tilde{q}_{Y\mid X}^{(k+1)}(y|x)} \\ 
&\overset{(c)}{=} \frac{p_{Y\mid X}(y|x)\tilde{q}_{X, Y}^{(k)}(x, y)^{1-\frac{1}{\alpha}}}{\tilde{q}_{Y\mid X}^{(k+1)}(y|x)}\cdot \frac{1}{\tilde{q}_{Y}^{(k)}(y)^{1-\frac{1}{\alpha}}}, \label{eq:LP_capacity_c} \\ 
&\sum_{x}p_{X}^{(k+1)}(x)^{\frac{\alpha}{2\alpha-1}} \left( \sum_{y} p_{Y\mid X}(y|x)r_{X\mid Y}^{(k)}(x|y)^{1-\frac{1}{\alpha}} \right)^{\frac{\alpha}{2\alpha-1}} \notag \\ 
&\overset{(d)}{=} \frac{p_{X}^{(k+1)}(x)^{\frac{\alpha}{2\alpha-1}} \left( \sum_{y} p_{Y\mid X}(y|x)r_{X\mid Y}^{(k)}(x|y)^{1-\frac{1}{\alpha}} \right)^{\frac{\alpha}{2\alpha-1}}}{\tilde{q}_{X}^{(k+1)}(x)} \\ 
&\overset{(e)}{=} \frac{p_{X}^{(k+1)}(x)^{\frac{\alpha}{2\alpha-1}}}{\tilde{q}_{X}^{(k+1)}(x)} \notag \\ 
&\times \left( \frac{p_{Y\mid X}(y|x)\tilde{q}_{X,Y}^{(k)}(x,y)^{1-\frac{1}{\alpha}}}{\tilde{q}_{Y\mid X}^{(k+1)}(y|x)}\cdot \frac{1}{\tilde{q}_{Y}^{(k)}(y)^{1-\frac{1}{\alpha}}} \right)^{\frac{\alpha}{2\alpha-1}} \\ 
&= \frac{p_{X}^{(k+1)}(x)^{\frac{\alpha}{2\alpha-1}}p_{Y\mid X}(y|x)^{\frac{\alpha}{2\alpha-1}}\tilde{q}_{X, Y}^{(k)}(x, y)^{\frac{\alpha-1}{2\alpha-1}}}{\tilde{q}_{X}^{(k+1)}(x)\tilde{q}_{Y\mid X}^{(k+1)}(y|x)^{\frac{\alpha}{2\alpha-1}}\tilde{q}_{Y}^{(k)}(y)^{\frac{\alpha-1}{2\alpha-1}}} \label{eq:denominator_LP_capacity}
\end{align}
for all $(x, y)\in \mathcal{X}\times \mathcal{Y}$, where $(a)$ follows from \eqref{eq:vc_LP_MI_01}, $(b)$ follows from \eqref{eq:LP_update_vc_tilde_q_X_Y} and \eqref{eq:LP_update_vc_hat_q_Y_X}, $(c)$ follows from \eqref{eq:LP_update_vc_r_X_Y}, $(d)$ follows from \eqref{eq:LP_update_vc_tilde_q_X_Y} and \eqref{eq:LP_update_vc_hat_q_X}, and $(e)$ follows from \eqref{eq:LP_capacity_c}.
We can substitute \eqref{eq:denominator_LP_capacity} into \eqref{eq:tilde_alpha_LP_capacity} to obtain 
\begin{align}
&F_{\alpha}^{\text{LP}}(p_{X}^{(k+1)}, \tilde{q}_{X, Y}^{(k+1)}, r_{X\mid Y}^{(k)}) \notag \\ 
&= \frac{2\alpha-1}{\alpha-1}\log \frac{p_{X}^{(k+1)}(x)^{\frac{\alpha}{2\alpha-1}}p_{Y\mid X}(y|x)^{\frac{\alpha}{2\alpha-1}}\tilde{q}_{X, Y}^{(k)}(x, y)^{\frac{\alpha-1}{2\alpha-1}}}{\tilde{q}_{X}^{(k+1)}(x)\tilde{q}_{Y\mid X}^{(k+1)}(y|x)^{\frac{\alpha}{2\alpha-1}}\tilde{q}_{Y}^{(k)}(y)^{\frac{\alpha-1}{2\alpha-1}}} \label{eq:finite_difference_LP_capacity}
\end{align}
for all $(x, y)\in \mathcal{X}\times \mathcal{Y}$. 
Then, we can use \eqref{eq:finite_difference_LP_capacity} to obtain the upper bound of the approximation error as follows:
\begin{align}
&C_{\alpha}^{\text{LP}} - \tilde{F}_{\alpha}^{\text{LP}}(p_{X}^{(k+1)}, \tilde{q}_{X, Y}^{(k+1)}, r_{X\mid Y}^{(k)}) \notag \\ 
&= \frac{\alpha}{1-\alpha} D(\tilde{q}_{X, Y}^{\text{opt}} || p_{X}^{\text{opt}}p_{Y\mid X}) + I(\tilde{q}_{X}^{\text{opt}}, \tilde{q}_{Y\mid X}^{\text{opt}}) \notag \\ 
&- \frac{2\alpha-1}{\alpha-1} \sum_{x, y}\tilde{q}_{X, Y}^{\text{opt}}(x, y) F_{\alpha}^{\text{LP}}(p_{X}^{(k+1)}, \tilde{q}_{X, Y}^{(k+1)}, r_{X\mid Y}^{(k)}) \\ 
&= \frac{\alpha}{1-\alpha} D(\tilde{q}_{X, Y}^{\text{opt}} || p_{X}^{\text{opt}}p_{Y\mid X}) + I(\tilde{q}_{X}^{\text{opt}}, \tilde{q}_{Y\mid X}^{\text{opt}}) -  \frac{2\alpha-1}{\alpha-1} \notag \\ 
&\times \sum_{x, y} \tilde{q}_{X, Y}^{\text{opt}}(x, y) \notag \\ 
&\times \log \frac{p_{X}^{(k+1)}(x)^{\frac{\alpha}{2\alpha-1}}p_{Y\mid X}(y|x)^{\frac{\alpha}{2\alpha-1}}\tilde{q}_{X, Y}^{(k)}(x, y)^{\frac{\alpha-1}{2\alpha-1}}}{\tilde{q}_{X}^{(k+1)}(x)\tilde{q}_{Y\mid X}^{(k+1)}(y|x)^{\frac{\alpha}{2\alpha-1}}\tilde{q}_{Y}^{(k)}(y)^{\frac{\alpha-1}{2\alpha-1}}} \\ 
&\overset{(f)}{=} \frac{\alpha}{1-\alpha}\sum_{x, y}\tilde{q}_{X, Y}^{\text{opt}}(x, y) 
\log \frac{\tilde{q}_{Y\mid X}^{\text{opt}}(y|x)^{\frac{1}{\alpha}}}{\tilde{q}_{Y\mid X}^{(k)}(y|x)^{\frac{1}{\alpha}}}\cdot \frac{\tilde{q}_{X}^{(k)}(x)\tilde{q}_{Y\mid X}^{(k)}(y|x)}{\tilde{q}_{X}^{(k+1)}(x)\tilde{q}_{Y\mid X}^{(k+1)}(y|x)} \notag \\ 
&\cdot \frac{\tilde{q}_{X}^{(k+1)}(x)^{\frac{1}{\alpha}}-1}{\tilde{q}_{X}^{(k)}(x)^{\frac{1}{\alpha}}-1}\cdot \frac{\tilde{q}_{Y}^{(k)}(y)^{\frac{1}{\alpha}-1}}{\tilde{q}_{Y}^{\text{opt}}(y)^{\frac{1}{\alpha}-1}} \\ 
&= \frac{1}{1-\alpha}\vE_{X}^{\tilde{q}_{X}^{\text{opt}}}\left[D(\tilde{q}_{Y\mid X}^{\text{opt}}(\cdot \mid X) || \tilde{q}_{Y\mid X}^{(k)}(\cdot \mid X))\right] \notag \\ 
&+ \frac{\alpha}{\alpha-1}\sum_{x, y}\tilde{q}_{X, Y}^{\text{opt}}(x,y)\log \frac{\tilde{q}_{X, Y}^{k+1}(x, y)}{\tilde{q}_{X, Y}^{(k)}(x, y)} \notag \\ 
&+ \sum_{x}\tilde{q}_{X}^{\text{opt}}(x) \log \frac{\tilde{q}_{X}^{(k+1)}(x)}{\tilde{q}_{X}^{(k)}(x)} - D(\tilde{q}_{Y}^{\text{opt}} || \tilde{q}_{Y}^{(k)}) \\ 
&\overset{(g)}{\leq} \frac{\alpha}{\alpha-1}\sum_{x, y}\tilde{q}_{X, Y}^{\text{opt}}\log \frac{\tilde{q}_{X, Y}^{k+1}(x, y)}{\tilde{q}_{X, Y}^{(k)}(x, y)} + \sum_{x}\tilde{q}_{X}^{\text{opt}}(x) \log \frac{\tilde{q}_{X}^{(k+1)}(x)}{\tilde{q}_{X}^{(k)}(x)}, 
\end{align}
where $(f)$ follows from \eqref{eq:update_p_X_LP}, i.e., $p_{X}^{\text{opt}} = \tilde{q}_{X}^{\text{opt}}$ and $p_{X}^{(k)}=\tilde{q}_{X}^{(k)}$,  
and $(g)$ follows from the nonnegativity of the Kullback--Leibler divergence and $\frac{1}{1-\alpha} < 0$ for $\alpha \in (1, \infty)$
\end{proof}

\section{Proof of Lemma \ref{lem:AC_upper_bound}} \label{proof:AC_upper_bound}
\begin{proof}
For convenience, let $\tilde{\gamma}^{(k)}(y):=\sum_{x}p_{X}(x)\tilde{q}_{Y\mid X}^{(k)}(y | x)$ and 
$\tilde{\gamma}^{\text{opt}}(y):=\sum_{x}p_{X}(x)\tilde{q}_{Y\mid X}^{\text{opt}}(y | x)$.
Then, Theorem \ref{thm:variational_characterization_Csiszar_MI} and Proposition \ref{prop:update_formulae_Csiszar_MI} can be used to obtain the following:  
\begin{align}
&\tilde{F}_{\alpha}^{\text{C}} (\tilde{q}_{Y\mid X}^{(k+1)}, r_{X\mid Y}^{(k)}) 
= \max_{\tilde{q}_{Y\mid X}} \tilde{F}_{\alpha}^{\text{C}}(\tilde{q}_{Y\mid X}, r_{X\mid Y}^{(k)}) \\
&\overset{(a)}{=} H(p_{X}) \notag \\
&+ \frac{\alpha}{\alpha-1}\sum_{x}p_{X}(x)\log \sum_{y}p_{Y\mid X}(y\mid x)r_{X\mid Y}^{(k)}(x\mid y)^{1-\frac{1}{\alpha}},  \label{eq:tilde_alpha_C} \\ 
&\sum_{y}p_{Y\mid X}(y\mid x)r_{X\mid Y}^{(k)}(x\mid y)^{1-\frac{1}{\alpha}} \notag \\
&\overset{(b)}{=} \frac{p_{Y\mid X}(y\mid x)r_{X\mid Y}^{(k)}(x\mid y)^{1-\frac{1}{\alpha}}}{\tilde{q}_{Y\mid X}^{(k+1)}(y\mid x)} \\ 
&\overset{(c)}{=} \frac{p_{Y\mid X}(y\mid x) \left( \frac{p_{X}(x)\tilde{q}_{Y\mid X}^{(k)}(y\mid x)}{\tilde{\gamma}^{(k)}(y)} \right)^{1-\frac{1}{\alpha}}}{\tilde{q}_{Y\mid X}^{(k+1)}(y\mid x)} \\ 
&= \frac{p_{Y\mid X}(y\mid x)p_{X}(x)^{1-\frac{1}{\alpha}}\tilde{q}_{Y\mid X}^{(k)}(y\mid x)^{1-\frac{1}{\alpha}}}{\tilde{q}_{Y\mid X}^{(k+1)}(y\mid x)}\times \frac{1}{\tilde{\gamma}^{(k)}(y)^{1-\frac{1}{\alpha}}}. \label{eq:denominator_C}
\end{align}
for all $(x, y)\in \mathcal{X}\times \mathcal{Y}$.
We can substitute \eqref{eq:denominator_C} into \eqref{eq:tilde_alpha_C} to obtain:
\begin{align}
&\tilde{F}_{\alpha}^{\text{C}} (\tilde{q}_{Y\mid X}^{(k+1)}, r_{X\mid Y}^{(k)}) = H(p_{X})+ \frac{\alpha}{\alpha-1} \sum_{x}p_{X}(x) \notag \\ 
& \log \frac{p_{Y\mid X}(y\mid x)p_{X}(x)^{1-\frac{1}{\alpha}}\tilde{q}_{Y\mid X}^{(k)}(y\mid x)^{1-\frac{1}{\alpha}}}{\tilde{q}_{Y\mid X}^{(k+1)}(y\mid x)}\cdot \frac{1}{\tilde{\gamma}^{(k)}(y)^{1-\frac{1}{\alpha}}} \\ 
&= \frac{\alpha}{\alpha-1}\sum_{x}p_{X}(x) \log \frac{\tilde{q}_{Y\mid X}^{(k)}(y\mid x)^{1-\frac{1}{\alpha}}}{\tilde{q}_{Y\mid X}^{(k+1)}(y\mid x)}\cdot \frac{p_{Y\mid X}(y\mid x)}{\tilde{\gamma}^{(k)}(y)^{1-\frac{1}{\alpha}}} \label{eq:finite_difference_C}
\end{align}
for all $y\in \mathcal{Y}$. 
Then, we can use \eqref{eq:finite_difference_C} to obtain the upper bound of the approximation error as follows:
\begin{align}
& I_{\alpha}^{\text{C}}(X; Y) - \tilde{F}_{\alpha}^{\text{C}}(\tilde{q}_{Y\mid X}^{(k+1)},  r_{X\mid Y}^{(k)}) \notag \\ 
&= I(p_{X}, \tilde{q}_{Y\mid X}^{\text{opt}}) + \frac{\alpha}{1-\alpha}D(p_{X}\tilde{q}_{Y\mid X}^{\text{opt}} || p_{X}p_{Y\mid X}) \notag \\
&\qquad \qquad - \sum_{y}\tilde{q}_{Y\mid X}^{\text{opt}}(y\mid x) \tilde{F}_{\alpha}^{\text{C}}(\tilde{q}_{Y\mid X}^{(k+1)},  r_{X\mid Y}^{(k)}) \\ 
&= \frac{\alpha}{1-\alpha}D(p_{X}\tilde{q}_{Y\mid X}^{\text{opt}} || p_{X}p_{Y\mid X}) + I(p_{X}, \tilde{q}_{Y\mid X}^{\text{opt}}) - \frac{\alpha}{\alpha-1} \notag \\ 
&\times \sum_{x, y}p_{X}(x)\tilde{q}_{Y\mid X}^{\text{opt}}(y\mid x) \log \frac{\tilde{q}_{Y\mid X}^{(k)}(y\mid x)^{1-\frac{1}{\alpha}}}{\tilde{q}_{Y\mid X}^{(k+1)}(y\mid x)}\cdot \frac{p_{Y\mid X}(y\mid x)}{\tilde{\gamma}^{(k)}(y)^{1-\frac{1}{\alpha}}} \\
&= \frac{\alpha}{\alpha-1}\sum_{x,y}p_{X}(x)\tilde{q}_{Y\mid X}^{\text{opt}}(y\mid x)\log \frac{\tilde{q}_{Y\mid X}^{(k+1)}(y\mid x)}{\tilde{q}_{Y\mid X}^{(k)}(y\mid x)} \notag \\ 
&+ \frac{1}{1-\alpha} \vE_{X}^{p_{X}}(x) \left[D(\tilde{q}_{Y\mid X}^{\text{opt}}(\cdot \mid X) || \tilde{q}_{Y\mid X}^{(k)}(\cdot \mid X))\right] - D(\tilde{\gamma}^{\text{opt}} || \tilde{\gamma}^{(k)}) \\ 
&\overset{(d)}{\leq} \frac{\alpha}{\alpha-1}\sum_{x,y}p_{X}(x)\tilde{q}_{Y\mid X}^{\text{opt}}(y\mid x)\log \frac{\tilde{q}_{Y\mid X}^{(k+1)}(y\mid x)}{\tilde{q}_{Y\mid X}^{(k)}(y\mid x)},
\end{align}
where $(d)$ follows from the nonnegativity of the Kullback--Leibler divergence and $\frac{1}{1-\alpha}<0$ for $\alpha\in (1, \infty)$.
\end{proof}

\section{Proof of Lemma \ref{lem:LP_upper_bound}} \label{proof:LP_upper_bound}
\begin{proof}
Let $\tilde{q}_{X, Y}^{(k)} = \tilde{q}_{X}^{(k)}\tilde{q}_{Y\mid X}^{(k)}$ and $\tilde{q}_{Y}^{(k)}(y) := \sum_{x}\tilde{q}_{X, Y}^{(k)}(x, y)$.
Then, from Theorem \ref{thm:variational_characterization_LP_MI} and Proposition \ref{prop:update_formulae_LP_MI_vc}, we can obtain the following:  
\begin{align}
&F_{\alpha}^{\text{LP}}(\tilde{q}_{X, Y}^{(k+1)}, q_{Y}^{(k)}) = \max_{\tilde{q}_{X, Y}}F_{\alpha}^{\text{LP}}(\tilde{q}_{X, Y}, q_{Y}^{(k)}) \notag \\
&\overset{(a)}{=} \frac{\alpha}{\alpha-1} \log \sum_{x} p_{X}(x) \left( \sum_{y}p_{Y\mid X}(y|x)^{\alpha}q_{Y}^{(k)}(y)^{1-\alpha} \right)^{\frac{1}{\alpha}}, \label{eq:tilde_alpha_LP} \\ 
&\sum_{y}p_{Y\mid X}(y|x)^{\alpha}q_{Y}^{k}(y)^{1-\alpha} \overset{(b)}{=} \frac{p_{Y\mid X}(y|x)^{\alpha}q_{Y}^{(k)}(y)^{1-\alpha}}{\tilde{q}_{Y\mid X}^{(k)}(y|x)} \\ 
&\overset{(c)}{=} \frac{p_{Y\mid X}(y|x)^{\alpha}\tilde{q}_{Y}^{(k)}(y)^{1-\alpha}}{\tilde{q}_{Y\mid X}^{(k)}(y|x)}, \label{eq:LP_c} \\ 
&\sum_{x}p_{X}(x)\left( \sum_{y}p_{Y\mid X}(y|x)^{\alpha}q_{Y}^{(k)}(y)^{1-\alpha} \right)^{\frac{1}{\alpha}} \notag \\ 
&\overset{(d)}{=} \frac{p_{X}(x) \left( \sum_{y} p_{Y\mid X}(y|x)^{\alpha}q_{Y}^{(k)}(y)^{1-\alpha} \right)^{\frac{1}{\alpha}}}{\tilde{q}_{X}^{(k+1)}(x)} \\ 
&\overset{(e)}{=} \frac{p_{X}(x) \left( \frac{p_{Y\mid X}(y|x)^{\alpha}\tilde{q}_{Y}^{(k)}(y)^{1-\alpha}}{\tilde{q}_{Y\mid X}^{(k)}(y|x)} \right)^{\frac{1}{\alpha}}}{\tilde{q}_{X}^{(k+1)}(x)} \\ 
&= \frac{p_{X}(x)p_{Y\mid X}(y|x)\tilde{q}_{Y}^{(k)}(y)^{\frac{1}{\alpha}-1}}{\tilde{q}_{X}^{(k+1)}(x)\tilde{q}_{Y\mid X}^{(k)}(y|x)^{\frac{1}{\alpha}}} \label{eq:denominator_LP}, 
\end{align}
where $(a)$ follows from Theorem \ref{thm:variational_characterization_LP_MI}, Proposition \ref{prop:update_formulae_LP_MI_vc}, and a simple calculation, 
$(b)$ follows from \eqref{eq:LP_update_tilde_q_X_Y} and \eqref{eq:LP_update_vc_bar_q_Y_X}; 
$(c)$ follows from \eqref{eq:LP_update_q_Y};  
$(d)$ follows from \eqref{eq:LP_update_tilde_q_X_Y} and \eqref{eq:LP_update_vc_bar_q_X};  
and $(e)$ follows from \eqref{eq:LP_c}. 

We can substitute \eqref{eq:denominator_LP} into \eqref{eq:tilde_alpha_LP} to obtain 
\begin{align}
&F_{\alpha}^{\text{LP}}(\tilde{q}_{X, Y}^{(k+1)}, q_{Y}^{(k)}) = \frac{\alpha}{\alpha-1} \log \frac{p_{X}(x)p_{Y\mid X}(y|x)\tilde{q}_{Y}^{(k)}(y)^{\frac{1}{\alpha}-1}}{\tilde{q}_{X}^{(k+1)}(x)\tilde{q}_{Y\mid X}^{(k)}(y|x)^{\frac{1}{\alpha}}} \label{eq:finite_difference_LP}
\end{align}
for all $(x, y) \in \mathcal{X}\times \mathcal{Y}$. 
Then, we can use \eqref{eq:finite_difference_LP} to obtain the upper bound of the approximation error as follows:
\begin{align}
&F_{\alpha}^{\text{LP}}(\tilde{q}_{X, Y}^{(k+1)}, q_{Y}^{(k)}) - I_{\alpha}^{\text{LP}}(X; Y) \notag \\ 
&= \sum_{x,y}\tilde{q}_{X, Y}^{\text{opt}}(x, y)F_{\alpha}^{\text{LP}}(\tilde{q}_{X, Y}^{(k+1)}, q_{Y}^{(k)}) \notag \\ 
&- \left( \frac{\alpha}{1-\alpha} D(\tilde{q}_{X, Y}^{\text{opt}} || p_{X}p_{Y\mid X}) + I(\tilde{q}_{X}^{\text{opt}}, \tilde{q}_{Y\mid X}^{\text{opt}}) \right) \\ 
&= D(\tilde{q}_{Y}^{\text{opt}} || \tilde{q}_{Y}^{(k)}) - \frac{1}{1-\alpha} D(\tilde{q}_{X, Y}^{\text{opt}} || \tilde{q}_{X}^{\text{opt}}\tilde{q}_{Y\mid X}^{(k)}) \notag \\ 
&- \frac{\alpha}{1-\alpha} D(\tilde{q}_{X}^{\text{opt}} || \tilde{q}_{X}^{(k+1)}) \\ 
&= \frac{1-2\alpha}{1-\alpha}D(\tilde{q}_{Y}^{\text{opt}} || \tilde{q}_{Y}^{(k)}) - \frac{1}{1-\alpha} D(\tilde{q}_{X, Y}^{\text{opt}} || \tilde{q}_{X}^{\text{opt}}\tilde{q}_{Y\mid X}^{(k)}) \notag \\
&+ \frac{\alpha}{1-\alpha} \left( D(\tilde{q}_{Y}^{\text{opt}} || \tilde{q}_{Y}^{(k)})-D(\tilde{q}_{X}^{\text{opt}} || \tilde{q}_{X}^{(k+1)}) \right) \\ 
&\overset{(f)}{\leq}  \frac{\alpha}{1-\alpha} \left( D(\tilde{q}_{Y}^{\text{opt}} || \tilde{q}_{Y}^{(k)})-D(\tilde{q}_{X}^{\text{opt}} || \tilde{q}_{X}^{(k+1)}) \right)  \\ 
&\overset{(g)}{\leq} \frac{\alpha}{1-\alpha} \left( D(\tilde{q}_{X}^{\text{opt}} || \tilde{q}_{X}^{(k)})-D(\tilde{q}_{X}^{\text{opt}} || \tilde{q}_{X}^{(k+1)}) \right) \\ 
&= \frac{\alpha}{1-\alpha} \sum_{x} \tilde{q}_{X}^{\text{opt}}(x) \log \frac{\tilde{q}_{X}^{(k+1)}(x)}{\tilde{q}_{X}^{(k)}(x)},
\end{align}
where $(f)$ follows from the nonnegativity of the Kullback--Leibler divergence and $\frac{1-2\alpha}{1-\alpha} \leq 0$ for $\alpha \in [1/2, 1]$ 
and $(g)$ follows from the data-processing inequality (also known as monotonicity) of the Kullback--Leibler divergence (see, e.g., \cite[Thm 2.17]{Polyanskiy_Wu_2024}).
\end{proof}


\end{document}